\documentclass[graybox]{svmult}
\usepackage{etex}

\usepackage[convert-footnotes = true]{notes2bib}

\usepackage{type1cm}        %

\usepackage{makeidx}         %
\usepackage{graphicx}        %
\usepackage{multicol}        %
\usepackage[bottom]{footmisc}%
\usepackage{url}

\usepackage[colorlinks=true, citecolor=blue, filecolor=black,linkcolor=red, urlcolor=blue, hypertexnames=false]{hyperref}

\usepackage{mathtools}
\usepackage{amsmath,amssymb,bm}
\usepackage{color}
\usepackage{mathrsfs}
\usepackage[letterpaper]{geometry}
\usepackage{multirow}
\usepackage{colortbl}
\usepackage{subfigure}

\usepackage{tikz}
\usetikzlibrary{calc}
\usepackage{pgfplots}

\usepackage[textsize=tiny,disable]{todonotes}

\makeindex             
\newcommand{\Ex}{\mathbb{E}}                                                   %
\newcommand{\Gauss}{\mathcal{N}}                                      %
\newcommand{\Unif}{\mathcal{U}}                                      %

\DeclareMathOperator\erf{erf}

\newcommand{\real}{\mathbb{R}}

\definecolor{nutsColor}{rgb}{0.9941,    0.7754,    0.6172}%
\definecolor{tmdrColor}{rgb}{0.5645,    0.7207,    0.8535}%
\definecolor{dramColor}{rgb}{0.84766,0.27734,0.00391}%

\newcommand{\jointPlotHeight}{5cm}
\newcommand{\jointPlotWidth}{5cm}

\newcommand{\rdvs}{x}               			 %
\newcommand{\tdvs}{z}               			 %
\newcommand{\tdv}{\boldsymbol{\tdvs}}                %
\newcommand{\rdv}{\boldsymbol{\rdvs}}                %
\newcommand{\tdva}{\tdv^*}                %
\newcommand{\rdva}{\rdv^*}                %
\newcommand{\tdvsa}{\tdvs^*}                %
\newcommand{\rdvsa}{\rdvs^*}                %

\newcommand{\trv}{\boldsymbol{Z}}                %
\newcommand{\sri}{\rdv_i}                %
\newcommand{\sti}{\tdv_i}                %

\newcommand{\pd}{n}                %

\newcommand{\data}{ {\boldsymbol{D}} }               %
\newcommand{\datai}{ \boldsymbol{d} }               %
\newcommand{\param}{\Theta}     %
\newcommand{\paramt}{\boldsymbol{\param}}    %
\newcommand{\ntd}{\pi}            %
\newcommand{\td}{\bar{\pi}}                  %
\newcommand{\rd}{\eta}                          %
\newcommand{\tmeas}{\mu_{\rm tar}}        %
\newcommand{\rmeas}{\mu_{\rm ref}}        %

\newcommand{\emap}{T}                          %
\newcommand{\demap}{{\nabla T}}                  %
\newcommand{\imap}{{\widetilde{T}}}       %
\newcommand{\dimap}{{\nabla \widetilde{T} }} %

\newcommand{\efmap}{S}                          %
\newcommand{\defmap}{{\nabla S}}                  %
\newcommand{\ifmap}{{\widetilde{S}}}       %

\newcommand{\mapp}{\gamma}                %
\newcommand{\Dkl}{\mathcal{D}_{\rm KL}}         %
\newcommand{\spaceMap}{\mathcal{T}}             %
\newcommand{\spaceMapT}{\mathcal{T}_\triangle}  %

\newcommand{\tmeasNew}{\tmeas} %
\newcommand{\rmeasNew}{\rmeas}  %
\newcommand{\ntda}{\widetilde{\pi}}                  %

\newcommand{\constrrefae}{, \quad(\rd-\rm{a.e.})  }  
\newcommand{\constrtarae}{, \quad(\ntd-\rm{a.e.})  }  
\newcommand{\pull}{^\sharp}  
\newcommand{\dimtot}{n} %
\newcommand{\rvtot}{\trv}  %
\newcommand{\frv}{\boldsymbol{D}}   %

\newcommand{\rdvf}{ \rdv_{\frv}}  
\newcommand{\rdvse}{ \rdv_{\srv}}  

\newcommand{\fdv}{{\boldsymbol{d}} } 
\newcommand{\fdvs}{d} 
\newcommand{\srv}{\boldsymbol{\Theta}}  %
\newcommand{\sdv}{{\boldsymbol{\theta}}}
\newcommand{\sdvs}{ \theta }  
\newcommand{\dimfrv}{n_{\fdvs}}  
\newcommand{\dimsrv}{n_{\sdvs}}
\newcommand{\jointfs}{\ntd_{\frv,\srv}}
\newcommand{\condsgf}{\ntd_{\srv\vert\frv}}
\newcommand{\margf}{\ntd_{\frv}}
\newcommand{\condsgfd}{\ntd_{\srv\vert\frv=\fdv}}
\newcommand{\rdsrv}{\rd_{\srv}}
\newcommand{\rdfrv}{\rd_{\frv}}

\newcommand{\sv}{\theta}  
\newcommand{\dimfv}{n_d}  
\newcommand{\dimsv}{n_\sv}  

\newcommand{\mapcf}{\emap^{\frv}}  
\newcommand{\mapcs}{\emap^{\srv}}  

\newcommand{\wb}{\boldsymbol{w}}  
\newcommand{\dvpl}{\wb}  
\newcommand{\solsystem}{\rdv_{\fdv}^\star}

\begin{document}

\title*{An introduction to sampling via measure transport}
\titlerunning{Sampling via measure transport}

\author{Youssef Marzouk, Tarek Moselhy, Matthew Parno, and Alessio Spantini}
\authorrunning{Marzouk, Moselhy, Parno, and Spantini}

\institute{Youssef Marzouk \at Massachusetts Institute of Technology, 77 Massachusetts Ave, Cambridge, MA, USA; \email{ymarz@mit.edu}
\and Tarek Moselhy \at D.\ E.\ Shaw Group, 1166 Avenue of the Americas, New York, NY, USA \email{tmoselhy@mit.edu}
\and Matthew Parno \at Massachusetts Institute of Technology, 77 Massachusetts Ave, Cambridge, MA, USA; \email{mparno@mit.edu}
\and Alessio Spantini \at Massachusetts Institute of Technology, 77 Massachusetts Ave, Cambridge, MA, USA; \email{spantini@mit.edu}}
\maketitle

\abstract{
We present the fundamentals of a measure transport approach to sampling. The idea is to construct a deterministic coupling---i.e., a transport map---between a complex ``target'' probability measure of interest and a simpler reference measure. Given a transport map, one can generate arbitrarily many independent and unweighted samples from the target simply by pushing forward reference samples through the map. If the map is endowed with a triangular structure, one can also easily generate samples from conditionals of the target measure. We consider two different and complementary scenarios: first, when only evaluations of the unnormalized target density are available, and second, when the target distribution is known only through a finite collection of samples. We show that in both settings the desired transports can be characterized as the solutions of variational problems.
We then address practical issues associated with the optimization--based construction of transports: choosing finite-dimensional parameterizations of the map, enforcing monotonicity, quantifying the error of approximate transports, and refining approximate transports by enriching the corresponding approximation spaces. %
Approximate transports can also be used to ``Gaussianize'' complex distributions and thus precondition conventional asymptotically exact sampling schemes. %
We place the measure transport approach in broader context, describing connections with other optimization--based samplers, with inference and density estimation schemes using optimal transport, and with alternative transformation--based approaches to simulation. We also sketch current work aimed at the construction of transport maps in high dimensions, exploiting essential features of the target distribution (e.g., conditional independence, low-rank structure). The approaches and algorithms presented here have direct applications to Bayesian computation and to broader problems of stochastic simulation.
}

\section{Introduction}
\label{sec:Intro}

Characterizing complex probability distributions is a fundamental and ubiquitous task in uncertainty quantification. In this context, the notion of ``complexity'' encompasses many possible challenges: non-Gaussian features, strong correlations and nonlinear dependencies, high dimensionality, the high computational cost of evaluating the (unnormalized) probability density associated with the distribution, or even intractability of the probability density altogether. 
Typically one wishes to characterize a distribution by evaluating its moments, or by computing the probability of an event of interest.
These goals can be cast as the computation of \emph{expectations} under the distribution, e.g., the computation of 
$\mathbb{E}[g(\trv)]$ where $g$ is some measurable function and $\trv$ is the random variable whose distribution we wish to characterize. 

The workhorse algorithms in this setting are sampling or ``simulation'' methods, of which the most broadly useful are Markov chain Monte Carlo (MCMC) \cite{Gilks1996,RobertBook2004,Brooks2011} or sequential Monte Carlo (SMC) \cite{smith2001sequential,Liu2004,delMoralSMC2006} approaches. 
{Direct} sampling from the distribution of interest---i.e., generating independent and unweighted samples---is typically impossible.
However, MCMC and SMC methods generate samples that can nonetheless be used to compute the desired expectations. In the case of MCMC, these samples are correlated, while in the case of importance sampling or SMC, the samples are endowed with weights. Non-zero correlations or non-uniform weights are in a sense the price to be paid for flexibility---for these approaches' ability to characterize arbitrary probability distributions. But if the correlations between successive MCMC samples decay too slowly, or if importance sampling weights become too non-uniform and the sample population thus degenerates, all these approaches become extremely inefficient. Accordingly, enormous efforts have been devoted to the design of improved MCMC and SMC samplers---schemes that generate more nearly independent or unweighted samples. While these efforts are too diverse to summarize easily, they often rest on the design of improved (and structure-exploiting) proposal mechanisms within the algorithms \cite{Haario2001,Andrieu2006,delMoralSMC2006,chorin2009implicit,Neal2011,Martin2012,Girolami2011,CuiDILI2016}.

As an alternative to the sampling approaches described above, we will consider \emph{transformations} of random variables, or perhaps more abstractly, \emph{transport maps} between probability measures. Let $\tmeas: \mathcal{B}(\mathbb{R}^n) \rightarrow \mathbb{R}_+$ be a probability measure that we wish to characterize, defined over the Borel
$\sigma$-algebra on $\mathbb{R}^n$, and let $\rmeas: \mathcal{B}(\mathbb{R}^n) \rightarrow \mathbb{R}_+$ be another probability measure from which we can easily generate independent and unweighted samples, e.g., a standard Gaussian. 
Then a transport map $T: \mathbb{R}^n \rightarrow \mathbb{R}^n$ \emph{pushes forward} $\rmeas$ to $\tmeas$ if and only if $\tmeas(A) = \rmeas(T^{-1}(A))$ for any set $A \in  \mathcal{B}(\mathbb{R}^n)$. We can write this compactly as 
\begin{equation}
T_\sharp \rmeas = \tmeas. \label{eq:measconst}
\end{equation}
In simpler terms, imagine generating samples $\sri \in \real^n$
that are distributed according to $\rmeas$ and then applying $T$ to each of these samples. Then the transformed samples 
$T(\sri)$ are distributed according to $\tmeas$. 

Setting aside for a moment questions of how to find such a transformation $T$ and what its properties might be, consider the significance of having $T$ in hand. First of all, given $T$ and the ability to sample directly from $\rmeas$, one can generate independent and unweighted samples from $\tmeas$ and from any of its marginal distributions. Moreover one can generate these samples cheaply, regardless of the cost of evaluating the probability density associated with $\tmeas$; with a map $T$ in hand, no further appeals to $\tmeas$ are needed. A transport map can also be used to devise deterministic sampling approaches, i.e., quadratures for non-standard measures $\tmeas$, based on quadratures for the reference measure $\rmeas$. Going further, if the transport map is endowed with an appropriate structure, it can enable direct simulation from particular conditionals of $\tmeas$. (We will describe this last point in more detail later.)

The potential to accomplish all of these tasks using measure transport is the launching point for this chapter. We will present a variational approach to the construction of transport maps, i.e., characterizing the desired maps as the solutions of particular optimization problems. We will also discuss 
the parameterization of transport maps
---a challenging task since maps are, in general, high-dimensional multivariate functions, which ultimately must be approximated in finite-dimensional spaces. Because maps are sought via optimization, standard tools for assessing convergence can be used. In particular, it will be useful to quantify the error incurred when the map is not exact, i.e., when we have only $T_\sharp \rmeas \approx \tmeas$, and to develop strategies for refining the map parameterization in order to reduce this error. More broadly, it will be useful to understand how the structure of the transport map (e.g., sparsity and other low-dimensional features) depends on the properties of the target distribution and how this structure can be exploited to construct and represent maps more efficiently. 

In discussing these issues, we will focus on two classes of map construction problems: (P1) constructing transport maps given the ability to evaluate only the unnormalized probability density of the target distribution, and (P2) constructing transport maps given only \textit{samples} from a distribution of interest, but no explicit density. These problems are frequently motivated by Bayesian inference, though their applicability is more general. In the Bayesian context, the first problem corresponds to the typical setup where one can evaluate the unnormalized density of the Bayesian posterior.  The second problem can arise when one has samples from the joint distribution of parameters and observations, and wishes to condition the former on a particular realization of the latter. This situation is related to approximate Bayesian computation (ABC) \cite{beaumont2002approximate,marin2012approximate}, and here our ultimate goal is conditional simulation. 

This chapter will present the basic formulations employed in the transport map framework and illustrate them with simple numerical examples. First, in Section~\ref{sec:transport}, we will recall foundational notions in measure transport and \emph{optimal transportation}. Then we will present variational formulations corresponding to Problems P1 and P2 above. In particular, Section~\ref{sec:DensConst} will explore the details of Problem P1: constructing transport maps from density evaluations. Section~\ref{sec:SampConst} will explore Problem P2: sample-based map construction. Section \ref{sec:MapParam} then discusses useful finite-dimensional parameterizations of transport maps. After presenting these formulations, we will describe connections between the transport map framework and other work in Section~\ref{sec:relatedwork}. In Section~\ref{sec:CondSamp} we describe how to simulate from certain conditionals of the target measure using a map and in Section \ref{sec:SampUses} we illustrate the framework on a Bayesian inference problem, including some comparisons with MCMC. Because this area is rapidly developing, this chapter will not attempt to capture all of the latest efforts and extensions. Instead, we will provide the reader with pointers to current work in Section \ref{sec:Conc}, along with a summary of important open issues.

\section{Transport maps and optimal transport}
\label{sec:transport}

A transport map $T$ satisfying \eqref{eq:measconst} can be understood as a deterministic coupling of two probability measures, and it is natural to ask under what conditions on $\rmeas$ and $\tmeas$ such a map exists. Consider, for instance, the case where $\rmeas$ has as an atom but $\tmeas$ does not; then there is no deterministic map that can push forward $\rmeas$ to $\tmeas$, since the probability contained in the atom cannot be split. Fortunately, the conditions for the existence of a map are quite weak---for instance, that $\rmeas$ be atomless \cite{Villani2009}. Henceforth we will assume that both the reference measure $\rmeas$ and the target measure $\tmeas$ are absolutely continuous with respect to the Lebesgue measure on $\mathbb{R}^n$, thus assuring the existence of transport maps satisfying \eqref{eq:measconst}.

There may be infinitely many such transformations, however. One way of choosing a particular map is to introduce a transport cost $c: \real^\pd\times\real^\pd\rightarrow \real$
such that  $c(\rdv,\tdv)$ represents the ``work'' needed to move a unit of mass from $\rdv$ to $\tdv$. 
The resulting cost of a particular map is then 
\begin{equation}
C(\emap) = \int_{\real^\pd} c \left (\rdv,\emap(\rdv) \right) \, {\rm d}\rmeas(\rdv). \label{eq:mapcost}
\end{equation}
Minimizing \eqref{eq:mapcost} while simultaneously satisfying \eqref{eq:measconst} corresponds to a problem first posed by Monge \cite{Monge1781} in 1781. The solution of this constrained minimization probem is the \emph{optimal} transport map.  Numerous properties of optimal transport have been studied in the centuries since. Of particular interest is the result of \cite{Brenier1991}, later extended by \cite{McCann1995}, which shows that when $c(\rdv,\emap(\rdv))$ is quadratic and
 $\rmeas$ is atomless, the optimal transport map exists and is unique; moreover this map is the gradient of a convex function and thus is monotone.  Generalizations of this result accounting for different cost functions and spaces can be found in \cite{Champion2011, Ambrosio2013, Feyel2004, Bernard2004}.  For a thorough contemporary development of optimal transport we refer to \cite{Villani2009,Villani2003}. 
The structure of the optimal transport map follows not only from the target and reference measures but also from the cost function in \eqref{eq:mapcost}.  For example, the quadratic cost of \cite{Brenier1991} and \cite{McCann1995} leads to maps that are in general dense, i.e., with each output of the map depending on every input to the map. However, if the cost is taken to be
\begin{equation}
  c(\rdv, \tdv) = \sum_{i=1}^{\pd} t^{i-1} | \rdvs_i- \tdvs_i |^2, \ t > 0,
  \label{eq:rosencost}
\end{equation}
then \cite{Carlier2010} and \cite{Bonnotte2013} show that the optimal map becomes lower triangular as $t\rightarrow 0$.  Lower triangular maps take the form
\begin{equation}  
\emap(\rdv) = \left[\begin{array}{l}
\emap^1(\rdvs_1)\\ 
\emap^2(\rdvs_1,\rdvs_2)\\ 
\vdots \\ 
\emap^\pd(\rdvs_1,\rdvs_2,\dots\rdvs_\pd)
\end{array}\right]\qquad \forall \rdv=(\rdvs_1,\ldots,\rdvs_\pd)\in\real^{\pd},
\label{eq:lowerTri}
\end{equation}
where $\emap^i$ represents output $i$ of the map.
 Importantly, when $\tmeas$ and $\rmeas$ are absolutely continuous, a unique lower triangular map satisfying \eqref{eq:measconst} exists; this map is exactly the Knothe-Rosenblatt rearrangement \cite{Rosenblatt1952, Carlier2010, Bonnotte2013}.

The numerical computation of the optimal transport map, for generic measures on $\mathbb{R}^n$, is a challenging task that is often restricted to very low dimensions 
\cite{benamou2000computational,angenent2003minimizing,loeper2005numerical,haber2010efficient}.
Fortunately, in the context of stochastic simulation and Bayesian inference, we are not particularly concerned with the optimality aspect of the transport; we just need to push forward the reference measure to the target measure. Thus we will focus on transports that are easy to compute, but that do not necessarily satisfy an optimality criterion based on transport cost. Triangular transports will thus be of particular interest to us. The triangular structure will make constructing transport maps feasible 
(see Sections \ref{sec:DensConst} and \ref{sec:SampConst}), 
conditional sampling straightforward (see Section \ref{sec:CondSamp}), 
and map inversion efficient (see Section \ref{sec:SampMap:inverse}). 
Accordingly, we will require the transport map to be (lower) triangular and search for a transformation that satisfies \eqref{eq:measconst}. The optimization problems arising from this formulation are described in the next two sections.

\section{Direct transport: constructing maps from unnormalized densities}
\label{sec:DensConst}

In this section we show how to construct a transport map that pushes forward a reference measure to the target measure when only evaluations of the \textit{unnormalized target density} are available. This is a central task in Bayesian inference, where the target is the posterior measure.

\subsection{Preliminaries}

We assume that both target and reference measures are absolutely continuous with respect to the Lebesgue measure on $\real^\pd$. Let $\ntd$ and $\rd$ be, respectively,  the normalized target and reference densities with respect to the Lebesgue measure. 
In what follows and for the sake of simplicity, we assume that both  $\ntd$ and $\rd$ are smooth strictly positive functions on their support. We seek a diffeomorphism $T$ (a smooth function with smooth inverse) that pushes forward the reference to the target measure,
\begin{equation} \label{eq:equalityMeasure}
 	   \tmeasNew = \rmeasNew \circ T^{-1},
\end{equation}
where $\circ$ denotes the composition of functions. In terms of densities, we will rewrite \eqref{eq:equalityMeasure} as $T_\sharp \rd = \ntd  $ over the support of the reference density. $T_\sharp \rd$ is the pushforward  of the reference density under the map $T$, and it is defined as:
\begin{equation}
  T_\sharp \,\rd  \coloneqq  \rd \circ T^{-1} \,  \vert \det \demap^{-1} \vert \, ,
\end{equation}
where $\demap^{-1}$ denotes the Jacobian of the inverse of the map. (Recall that the Jacobian determinant  $\det \demap^{-1}$ is equal to $1 /  \left ( \det \demap \circ \emap^{-1} \right ) $.) 
As noted in the introduction, if $(\sri)_i$ are independent samples from $\rd$, then $( T(\sri) )_i$ are independent samples from  $T_\sharp \rd$. 
(Here and throughout the chapter, we use the notation $(\sri)_i$ as a shorthand for $(\sri)_{i=1}^M$  to denote a collection $(\rdv_1,\ldots,\rdv_M)$ whenever the definition of the cardinality $M$ is either unimportant or possibly infinite.)  Hence, if we find a transport map $T$  that satisfies $T_\sharp \rd = \ntd  $, then $( T(\sri) )_i$ will be independent samples from the target distribution.
In particular,  the change of variables formula:
\begin{equation}
	\int g(\rdv) \,\ntd(\rdv) \, {\rm d}\rdv = \int [g\circ T](\rdv) \,\rd(\rdv) \, {\rm d}\rdv 
\end{equation}
holds for any integrable real-valued function $g$ on $\mathbb{R}^n$ \cite{adams1996measure}. The map therefore allows for direct computation of posterior expectations.

\subsection{Optimization problems}

Now we describe a series of optimization problems whose solution yields the desired transport map. Let $\Dkl (\, \pi_1 \, \Vert \, \pi_2) $ denote the Kullback-Leibler (K--L) divergence from a probability measure with density $\pi_1$ to a probability measure with density $\pi_2$, i.e.,  
$$
\Dkl (\, \pi_1 \, || \, \pi_2 \, ) = \mathbb{E}_{\pi_1} \left ( \log \frac{\pi_1}{\pi_2} \right ) \, ,
$$
and let $\spaceMap$ be an appropriate set of diffeomorphisms. %
Then, any 
global 
minimizer of the optimization problem:
\begin{eqnarray}  \label{abstractOptimDirect}
			& {\rm min}     &  \Dkl (\, T_\sharp \rd \,||\, \ntd \,)  \\
			& {\rm s.t.}      &    \det \demap > 0\constrrefae    \nonumber  \\
	        &  				    & T \in \spaceMap       \nonumber 
\end{eqnarray}
is a valid  transport map that pushes forward the reference to the target measure \footnote{See \cite{Moselhy2011} for a discussion on the asymptotic equivalence of the K--L divergence and Hellinger distance in the context of transport maps.}.
In fact, any 
global 
minimizer of \eqref{abstractOptimDirect} achieves the minimum cost $\Dkl (\, T_\sharp \rd \,||\, \ntd \,)=0$ and implies that $T_\sharp \rd = \ntd$.
The constraint $\det \demap > 0 $ ensures that the pushforward density $T_\sharp \rd$ is strictly positive on the support of the target.
In particular, the constraint $\det \demap > 0 $ ensures that the K--L divergence evaluates to finite values over $\spaceMap$ and does not rule out any useful transport map since we assume that both target and reference densities are positive.
The existence of global minimizers of \eqref{abstractOptimDirect} is a standard result in the theory of deterministic couplings between random variables \cite{Villani2009}.

Among these minimizers, a particularly useful map is given by the Knothe-Rosenblatt rearrangement \cite{Carlier2010}.
In our hypothesis, the Knothe-Rosenblatt rearrangement is a \textit{triangular} (in the sense that the $k$th component of the map depends only on the first $k$ input variables) diffeomorphism $T$ such that $\demap \succ 0$.  That is, each eigenvalue of $\demap$ is real and positive. Thus it holds that $\det \demap >0$.
Notice that for a triangular map the eigenvalues of $\demap$ are just the diagonal  entries of this matrix.
The Knothe-Rosenblatt rearrangement is also monotone increasing according to the
lexicographic order on $\real^\pd$ \footnote{The lexicographic order on $\real^\pd$ is defined as follows. For $x,y \in \real^\pd$, we define $x \preceq y$ if and only if either $x = y$ or the first nonzero coordinate in $y-x$ is positive \cite{ghorpade2010course}. $\preceq$ is a total order  on $\real^\pd$. Thus, we  define $T$ to be a monotone increasing function if and only if  $x \preceq y$ implies $T(x) \preceq T(y)$. Notice that monotonicity can be defined with respect to any order on $\real^\pd$ (e.g., $\preceq$ need not be the lexicographic order). 
There is no natural order on $\real^\pd$ except when $\pd=1$. It is easy to verify that for a triangular function $T$, monotonicity with respect to the lexicographic order is equivalent to the following: the $k$th component of $T$ is a monotone function of the $k$th input variable.}.

It turns out that we can further constrain \eqref{abstractOptimDirect} so that the
Knothe-Rosenblatt rearrangement is the \textit{unique} 
global 
minimizer of: 
\begin{eqnarray}  \label{abstractOptimDirectRosenblatt}
			& {\rm min}     &  \Dkl (\, T_\sharp \rd \,||\, \ntd \,)  \\
			& {\rm s.t.}      &    \demap \succ 0\constrrefae       \nonumber  \\
	        &  				    & T \in \spaceMapT       \nonumber 
\end{eqnarray}
where $\spaceMapT$  is now the vector space of smooth triangular maps. 
The constraint $\demap \succ 0$ suffices to enforce invertibility of a feasible triangular map.
\eqref{abstractOptimDirectRosenblatt} is a far better behaved optimization problem than the original formulation \eqref{abstractOptimDirect}. 
Hence, for the rest of this section we will focus on the computation of a  Knothe-Rosenblatt rearrangement by solving \eqref{abstractOptimDirectRosenblatt}.
Recall that our goal is just to compute a transport map from $\rmeasNew$ to $\tmeasNew$. If there are multiple transports, we can opt for the easiest one to compute. 
A possible drawback of a triangular transport is that the complexity of a parameterization of the map depends on the ordering of the input variables. This dependence motivates questions of what is the ``best'' ordering, or how to find at least a ``good'' ordering. 
We refer the reader to \cite{spantini16markov} for an in-depth discussion of this topic. For the computation of general non-triangular transports see \cite{Moselhy2011}; for a generalization of the framework to compositions of maps see \cite{Moselhy2011,Parno2014thesis}; and for the computation of optimal transports see, for instance, \cite{Villani2009,benamou2000computational,angenent2003minimizing,loeper2005numerical}.

Now let $\td$ denote any \textit{unnormalized} version of the target density. For any map $T$ in the feasible set of \eqref{abstractOptimDirectRosenblatt}, the objective function can be written as:
\begin{eqnarray} \label{KLdivergence}
	\Dkl (\, T_\sharp \rd \,||\, \ntd \,) & = & \Dkl (\, \rd \,||\, T^{-1}_\sharp \ntd \,) \\
	                                           &  = & \Ex_{\rd} [ - \log \td \circ T - \log \det \demap   ] \,  +\, \mathfrak{C} ,\nonumber  
\end{eqnarray}
where $\Ex_{\rd}[\cdot]$ denotes integration with respect to the reference measure, and  $\mathfrak{C}$ is a term independent of the transport map and thus a constant for the purposes of optimization. (In this case, $\mathfrak{C}=\log \beta + \log \eta$, where $\beta \coloneqq \td/\ntd$ is the
normalizing constant of the target density.) The resulting optimization problem reads as:
\begin{eqnarray}  \label{OptimDirect}
			& {\rm min}     &  \Ex_{\rd} [ - \log \td \circ T - \log \det \demap   ]   \\
			& {\rm s.t.}      &    \demap \succ 0\constrrefae       \nonumber  \\
	        &  				    & T \in \spaceMapT       \nonumber 
\end{eqnarray}
Notice that we can evaluate the objective of \eqref{OptimDirect} given only the {unnormalized} density $\td$ and a way to compute the
integral $\Ex_{\rd}[\cdot]$.  There exist a host of techniques to approximate the integral with respect to the reference measure, including quadrature and cubature formulas, sparse quadratures, Monte Carlo methods, and quasi-Monte Carlo (QMC) methods. The choice between these methods is typically dictated by the dimension of the reference space. In any case, the reference measure is usually chosen so that the integral with respect to $\rd$ can be approximated easily and accurately. 
For instance, if $\rmeasNew$ is a standard Gaussian measure, then we can generate arbitrarily many independent samples to yield an approximation of  $\Ex_{\rd}[\cdot]$ to any desired accuracy. This will be a crucial difference relative to the sample--based  construction of the map described in  Section~\ref{sec:SampConst},  where samples from the target distribution are required to accurately solve the corresponding optimization problem.

\eqref{OptimDirect} is a linearly constrained nonlinear differentiable optimization problem. It is also non-convex unless, for instance, the target density is log-concave \cite{Kim2013}. That said, many statistical models have log-concave posterior distributions and hence yield convex map optimization problems; consider, for example, a log-Gaussian Cox process \cite{Girolami2011}. All $n$ components of the map have to be computed simultaneously, and each evaluation of the objective function of \eqref{OptimDirect} requires an evaluation of the unnormalized target density. The latter is also the minimum requirement for alternative sampling techniques such as MCMC. Of course, the use of derivatives in the context of the optimization problem is crucial for computational efficiency, especially for high-dimensional parameterizations of the map. The same can be said of the state-of-the-art MCMC algorithms that use gradient or Hessian information from the log-target density to yield better proposal distributions, such as Langevin or Hamiltonian MCMC \cite{Girolami2011}. In the present context, we advocate the use of quasi-Newton (e.g., BFGS) or Newton methods \cite{wright1999numerical}
 to solve \eqref{OptimDirect}. These methods must be paired with a finite-dimensional parameterization of the map; in other words, we must solve \eqref{OptimDirect}  over a finite-dimensional space  $\spaceMapT^h \subset \spaceMapT$ of triangular diffeomorphisms. In Section \ref{sec:MapParam} we will discuss various choices for  $\spaceMapT^h$ along with ways of enforcing the monotonicity constraint $\demap \succ 0$.

\subsection{Convergence, bias, and approximate maps}
\label{s:bias}
A transport map provides a deterministic solution to the problem of sampling from a given unnormalized density, avoiding classical stochastic tools such as MCMC.  Once the transport map is computed, we can quickly generate independent and unweighted samples from the target distribution without further evaluating the target density \cite{Moselhy2011}.  This is a major difference with respect to MCMC.  Of course, the density--based transport map framework essentially exchanges a challenging sampling task for a challenging optimization problem involving function approximation. Yet there are several advantages to dealing with an optimization problem rather than a sampling problem. Not only can we rely on a rich literature of robust algorithms for the solution of high-dimensional nonlinear optimization problems, but we also inherit the notion of convergence criteria. The latter point is crucial. 

A major concern in MCMC sampling methods is the lack of clear and generally applicable convergence criteria. It is a non-trivial task to assess the stationarity of an ergodic Markov chain, let alone to measure the distance between the empirical measure given by the MCMC samples and the target distribution \cite{GorhamMa15}. In the transport map framework, on the other hand, the convergence criterion is borrowed directly from standard optimization theory \cite{luenberger1968optimization}. As shown in \cite{Moselhy2011}, the K--L divergence $\Dkl (\, T_\sharp \rd \,||\, \ntd \,)$ can be estimated as:
	\begin{equation} \label{varianceConvergence}
	 \Dkl (\, T_\sharp \rd \,||\, \ntd \,) \approx \frac{1}{2} \,\mathbb{V}\mathrm{ar}_{\rd}[ \log \rd - \log T^{-1}_\sharp \td  ]
	\end{equation}	 
up to second-order terms in the limit of $\mathbb{V}\mathrm{ar}_{\rd}[ \log \rd - \log T^{-1}_\sharp \td  ]\rightarrow 0$,  even if the normalizing constant of the target density is unknown. (Notice that \eqref{varianceConvergence} contains only the unnormalized target density $\td$.) Thus one can monitor \eqref{varianceConvergence} to estimate the divergence between the pushforward of a given map and the desired target distribution. Moreover, the transport map algorithm also provides an estimate of the normalizing constant $\beta \coloneqq \td/\ntd$ of the target density as  \cite{Moselhy2011}:
\begin{equation}
	\beta = \exp \Ex_{\rd} [ \log \rd - \log T^{-1}_\sharp \td ] \, .
\label{eq:normconst}
\end{equation}
The normalizing constant is a particularly useful quantity in the context of Bayesian model selection \cite{Gelman2003}. Reliably retrieving this normalizing constant from MCMC samples requires additional effort (e.g.,  \cite{gelman1998simulating,chib2001marginal}). 
The numerical solution of \eqref{OptimDirect} entails at least two different approximations. First, the infinite-dimensional function space $\spaceMapT$ must be replaced with a finite dimensional subspace $\spaceMapT^h \subset \spaceMapT$.  For example, each component of the map can be approximated in a total-degree polynomial space, as discussed in 
Section \ref{sec:MapParam}. Let $h$ parameterize a sequence of possibly nested finite-dimensional approximation spaces $(\spaceMapT^h)_h$. Then, as the dimension of $\spaceMapT^h$ grows, we can represent increasingly complex maps.
Second, the expectation with respect to the reference measure in the objective of \eqref{OptimDirect} must also be approximated. As discussed earlier, one may take any of several approaches. As a concrete example, consider approximating $\Ex_{\rd}[\cdot]$ by a Monte Carlo sum with $M$ independent samples $(\sri)_i$ from the reference measure. Clearly, as the cardinality of the sample set grows, the approximation of $\Ex_{\rd}[\cdot]$ becomes increasingly accurate. An instance of an approximation of \eqref{OptimDirect} is then:
\begin{eqnarray}  \label{OptimDirectApprox}
			& {\rm min}     &  \frac{1}{M} \sum_{i=1}^M \,\left( - \log \td( T(\sri) ) - 
			\sum_{k=1}^\pd \log\, \partial_k T^k(\sri) \right)  \\
			& {\rm s.t.}      &   \partial_k T^k > 0,  \ \  
			k=1,\ldots,\pd\constrrefae     \nonumber  \\
	        &  				    & T \in \spaceMapT^h \subset \spaceMapT      \nonumber 
\end{eqnarray}
where we have simplified the monotonicity constraint  $\demap \succ 0$ by using the fact that $\demap$ is lower triangular for maps in $\spaceMapT$. Above we require that the monotonicity constraint be satisfied over the whole support of the reference density. We will discuss ways of strictly guaranteeing such a property in Section~\ref{sec:MapParam}. Depending on the parameterization of the map, however, the monotonicity constraint is sometimes relaxed (e.g., in the tails of $\rd$); doing so comprises a third source of approximation. 

Given a fixed sample set, \eqref{OptimDirectApprox} is a sample-average approximation (SAA) 
\cite{Kleywegt2002} 
of \eqref{OptimDirect}, to which we can apply standard numerical optimization techniques \cite{wright1999numerical}. 
Alternatively, one can regard \eqref{OptimDirect} as a stochastic program and solve it using stochastic approximation techniques \cite{Kushner2003, spall2005introduction}. In either case, the transport map framework allows efficient \textit{global} exploration of the parameter space via optimization. The exploration is global since with a transport map $\emap$ 
we are essentially trying to push forward the entire collection of reference samples $(\sri)_i$ to samples $(T(\sri))_i$ that fit the entire target distribution. Additionally, we can interpret the finite-dimensional parameterization of a candidate transport as a constraint on the relative motion of the pushforward particles $(T(\sri))_i$.

The discretization of the integral $\Ex_{\rd}[\cdot]$ in \eqref{OptimDirectApprox} reveals another important distinction of the transport map framework from MCMC methods. At every optimization iteration, we need to evaluate the target density $M$ times. If we want an accurate approximation of $\Ex_{\rd}[\cdot]$, then $M$ can be large. But these $M$ evaluations of the target density can be performed in an embarrassingly \textit{parallel} manner. %
This is a fundamental difference from standard MCMC, where the evaluations of the target density are inherently sequential. (For exceptions to this paradigm, see, e.g., \cite{calderhead2014general}.)

A minimizer of \eqref{OptimDirectApprox} is an approximate transport map $\imap$; this map can be written as $\imap(\,\cdot\,;\,M,h)$ to reflect dependence on the approximation parameters $(M,h)$ defined above. Ideally, we would like the approximate map $\imap$ to be as close as possible to the true minimizer $T \in \spaceMapT$ of \eqref{OptimDirect}.  Yet it is also important to understand the potential of an approximate map alone.
If $\imap$ is not the exact transport map, then $\imap_\sharp \rd$ will not be the target density. The pushforward density $\imap_\sharp \rd$ instead defines an approximate target: $\ntda \coloneqq \imap_\sharp \rd$. We can easily sample from $\ntda$ by pushing forward reference samples through $\imap$. If we are interested in estimating integrals of the form $\Ex_{\ntd}[g]$ for some integrable function $g$, then we can try to use Monte Carlo estimators of  $\Ex_{\ntda}[g]$ to approximate $\Ex_{\ntd}[g]$. This procedure will result in a biased estimator for $\Ex_{\ntd}[g]$ since $\ntda$ is not the target density. It turns out that this bias can be bounded as:
\begin{equation} \label{biasBound}
 \left \Vert  \Ex_{\ntd}[g] - \Ex_{\ntda}[g]  \right \Vert \le \, \mathcal{C}(g,\ntd,\ntda) \,\,\sqrt{ \Dkl (\, T_\sharp \rd \,||\,  \ntd \,) }
\end{equation}
where $\mathcal{C}(g,\ntd,\ntda) \coloneqq \sqrt{2} \left ( \Ex_{\ntd}[\|g\|^2] + \Ex_{\ntda}[\|g\|^2] \right )^{\frac{1}{2}}$.  The proof of this result is in \cite[Lemma 6.37]{stuart2010inverse} together with a similar result for the approximation of the second moments. Note that the K--L divergence on the right-hand side of \eqref{biasBound} is exactly the quantity we minimize in \eqref{OptimDirect} during the computation of a transport map, and it can easily be estimated using \eqref{varianceConvergence}.  Thus the transport map framework allows a systematic control of the bias resulting from estimation of $\Ex_{\ntd}[g]$ by means of an approximate map $\imap$. 

In practice, the mean-square error in approximating $\Ex_{\ntd}[g]$ using $\Ex_{\ntda}[g]$ will be entirely due to the bias described in \eqref{biasBound}. The reason is that a Monte Carlo estimator of $\Ex_{\ntda}[g]$ can be constructed to have virtually no variance: one can cheaply generate an arbitrary number of independent and unweighted samples from $\ntda$ using the approximate map. Hence the approximate transport map yields an essentially zero-variance biased estimator of the quantity of interest $\Ex_{\ntd}[g]$.  This property should be contrasted with MCMC methods which, while asymptotically unbiased, yield estimators of $\Ex_{\ntd}[g]$  that have nonzero variance and bias for finite sample size.

If one is not content with the bias associated with an approximate map, then there are at least two ways to proceed. First, one can simply refine the approximation parameters $(M,h)$ to improve the current approximation of the transport map. On the other hand, it is straightforward to apply any classical sampling technique (e.g., MCMC, importance sampling) to the \textit{pullback} density $\imap\pull\ntd = \imap^{-1}_\sharp \ntd$. This density is defined as
\begin{equation}
 \imap\pull \ntd \coloneqq \ntd\circ \imap \, \vert \det \dimap \vert \, ,
\label{eq:pullbackdensity}
\end{equation}
and can be evaluated (up to a normalizing constant) at no significant additional cost compared to the original target density $\ntd$. If $(\sri)_i$ are samples from the pullback 
$\imap\pull \ntd$, then $(\imap(\sri))_i$ will be samples from the \textit{exact} target distribution. But there are clear advantages to sampling $\imap\pull \ntd$ instead of the original target distribution. 
Consider the following: if $\imap$ were the exact transport map, then 
$\imap\pull \ntd$ would simply be the reference density. 
With an approximate map, we still expect $\imap\pull \ntd$ to be close to the reference density---more precisely, closer to the reference (in the sense of K--L divergence) than was the original target distribution. In particular, when the reference is a standard Gaussian, the pullback will be closer to a standard Gaussian than the original target. Pulling back through an approximate map thus ``Gaussianizes'' the target, and can remove the correlations and nonlinear dependencies that make sampling a challenging task. In this sense, we can interpret an approximate map as a general \textit{preconditioner} for any known sampling scheme. See \cite{Parno2015mcmc} for a full development of this idea in the context of MCMC.

There is clearly a trade-off between the computational cost associated with constructing a more accurate transport map and the costs of ``correcting'' an approximate map by applying an exact sampling scheme to the pullback. Focusing on the former, it is natural to ask how to refine the finite-dimensional approximation space $\spaceMapT^h$ so that it can better capture the true map $T$. Depending on the problem at hand, a na\"{i}ve finite-dimensional parameterization of the map might require a very large number of degrees of freedom before reaching a satisfactory approximation. This is particularly true when the parameter dimension $n$ is large, and is a challenge shared by any function approximation algorithm (e.g., high-dimensional regression). We will revisit this issue in Section~\ref{sec:Conc}, but the essential way forward is to realize that the structure of the target distribution is reflected in the structure of the transport map. For instance, conditional independence in the target distribution yields certain kinds of \textit{sparsity} in the transport, which can be exploited when solving the optimization problems above. Many Bayesian inference problems also contain low-rank structure that causes the map to depart from the identity only on low-dimensional subspace of $\mathbb{R}^n$. From the optimization perspective, adaptivity can also be driven via a systematic analysis of the first variation of the K--L divergence $ \Dkl (\, T_\sharp \rd \,||\, \ntd \,)$ as a function of $T\in\spaceMap$.

\section{Inverse transport: constructing maps from samples}
\label{sec:SampConst}

In the previous section we focused on the computation of a transport
map that pushes forward a reference measure to a target measure, in
settings where the target density can be evaluated up to a normalizing
constant and where it is simple to approximate integrals with respect
to the reference measure (e.g., using quadrature, Monte Carlo, or
QMC). In many problems of interest, including density destimation
\cite{Tabak2013} and approximate Bayesian computation, however, it is
not possible to evaluate the unnormalized target density $\td$.

In this section we assume that the target density is unknown and that
we are only given a finite number of samples distributed according to
the target measure.  We show that under these hypotheses it is
possible to efficiently compute an \emph{inverse transport}---a
transport map that pushes forward the target to the reference
measure---via convex optimization.  The direct transport---a transport
map that pushes forward the reference to the target measure---can then
be easily recovered by inverting the inverse transport map pointwise,
taking advantage of the map's triangular structure.

\subsection{Optimization problem}
We denote the inverse transport by
$\efmap:\real^n \rightarrow \real^n$ and again assume that the
reference and target measures are absolutely continuous with respect
to the Lebesgue measure on $\real^n$, with smooth and positive
densities.  The inverse transport pushes forward the target to the
reference measure:
\begin{equation}
	\rmeas = \tmeas\circ \efmap^{-1}.
\end{equation}
We focus on the inverse triangular transport because it can be
computed via convex optimization given samples from the target
distribution. It is easy to see that the monotone increasing
Knothe-Rosenblatt rearrangement that pushes forward $\tmeas$ to
$\rmeas$ is the unique minimizer of
\begin{eqnarray}  \label{abstractOptimInverseRosenblatt}
			& {\rm min}     &  \Dkl (\, \efmap_\sharp \ntd \,||\, \rd \,)  \\
			& {\rm s.t.}      &    \defmap \succ 0\constrtarae      \nonumber  \\
	        &  				    &  \efmap \in \spaceMapT       \nonumber 
\end{eqnarray}
where $\spaceMapT$ is the space of smooth  triangular maps. If $\efmap$ is a minimizer of \eqref{abstractOptimInverseRosenblatt}, then 
$ \Dkl (\, \efmap_\sharp \ntd \,||\, \rd \,)=0$ and thus $ \efmap_\sharp \ntd = \rd$.
For any map $\efmap$ in the feasible set of \eqref{abstractOptimInverseRosenblatt}, the
objective function can be written as:
\begin{eqnarray} \label{KLdivergenceInverse}
	\Dkl (\, \efmap_\sharp \ntd \,||\, \rd \,) 
	& = & \Dkl (\,  \ntd \,||\, \efmap^{-1}_\sharp \rd \,) \\
	&  = & \Ex_{\ntd} [ - \log \rd \circ \efmap - \log \det \defmap    ]  + \mathfrak{C}
	\nonumber  
\end{eqnarray}
where $\Ex_{\ntd}[\cdot]$ denotes integration with respect to the
target measure, and where 
$\mathfrak{C}$ is once again a term independent of the transport map and thus a constant for the purposes of optimization.
The resulting optimization problem is a stochastic program given by:
\begin{eqnarray}  \label{OptimInverse}
			& {\rm min}     &  \Ex_{\ntd} [ - \log \rd \circ \efmap - \log \det \defmap   ]   \\
			& {\rm s.t.}      &    \defmap \succ 0\constrtarae      \nonumber  \\
	        &  				    & \efmap \in \spaceMapT       \nonumber 
\end{eqnarray}
Notice that \eqref{OptimInverse} is equivalent to \eqref{OptimDirect}
if we interchange the roles of target and reference densities.
Indeed, the K--L divergence is not a symmetric function.  The
direction of the K--L divergence, i.e.,
$\Dkl (\, \efmap_\sharp \ntd \,||\, \rd \,)$ versus
$\Dkl (\, T_\sharp \rd \,||\, \ntd \,)$, is one of the key
distinctions between the sample-based map construction 
presented in this section and the density-based construction of
Section \ref{sec:DensConst}.  The choice of direction in
\eqref{KLdivergenceInverse} involves integration over the target
distribution, as in the objective function of \eqref{OptimInverse},
which we approximate using the given samples.
Let $(\sti)_{i=1}^M$ be $M$ samples from the target distribution.
Then, a sample-average approximation (SAA) \cite{Kleywegt2002} of
\eqref{OptimInverse} is given by:
\begin{eqnarray}  \label{OptimInverseSAA}
			& {\rm min}     &  \frac{1}{M}\sum_{i=1}^M
			  - \log \rd(\efmap(\sti)) - \log \det \defmap (\sti)   \\
			& {\rm s.t.}      &    
			\partial_k \efmap^k > 0 \qquad k=1,\ldots,n\constrtarae      \nonumber  \\
	        &  				    & \efmap \in \spaceMapT       \nonumber 
\end{eqnarray}
where we use the lower triangular structure of $\defmap$ to rewrite
the monotonicity constraint, $ \defmap \succ 0$, as a sequence of
essentially one dimensional monotonicity constraints:
$\partial_k \efmap^k > 0$  for $k=1,\ldots,n$.  We
note, in passing, that the monotonicity constraint can be
satisfied automatically by using monotone parameterizations of the triangular
transport (see Section \ref{s:mapMonotone}).  One can certainly use stochastic
programming techniques to solve \eqref{OptimInverse} depending on the
availability of target samples (e.g., stochastic approximation
\cite{Kushner2003,spall2005introduction}).  SAA, on the other hand, turns
\eqref{OptimInverse} into a deterministic optimization problem and
does not require generating new samples from the target distribution,
which could involve running additional expensive simulations or performing
new experiments.  Thus, SAA is generally the method of choice to solve
\eqref{OptimInverse}.  However, stochastic approximation may be better
suited for applications involving streaming data or massive sample
sets requiring single pass algorithms.

\subsection{Convexity and separability of the optimization problem}
\label{subsec:convex}
Note that \eqref{OptimInverseSAA} is a convex optimization problem
as long as the reference density is log--concave \cite{Kim2013}. Since
the reference density is a degree of freedom of the problem, it can
always be chosen to be log--concave.  Thus, \eqref{OptimInverseSAA}
can be a convex optimization problem \emph{regardless} of the
particular structure of the target.  This is a major difference from
the density-based construction of the map in Section
\ref{sec:DensConst}, where the corresponding optimization problem
\eqref{OptimDirect} is convex only under certain conditions on the
target (e.g., that the target density be log--concave).

For smooth reference densities, the objective function of
\eqref{OptimInverseSAA} is also smooth. Moreover, its gradients do not
involve derivatives of the log--target density---which might require
expensive adjoint calculations when the target density contains a PDE
model.  Indeed, the objective function of \eqref{OptimInverseSAA} does
not contain the target density at all! This feature should be contrasted
with the density-based construction of the map, where the objective
function of \eqref{OptimDirect} depends explicitly on the log-target
density.  Moreover, if the reference density can be written as the
product of its marginals, then \eqref{OptimInverseSAA} is a separable
optimization problem, i.e., each component of the inverse transport
can be computed independently and in parallel.

As a concrete example, let the reference measure be standard
Gaussian. In this case, \eqref{OptimInverseSAA} can be written as
\begin{eqnarray}  \label{OptimInverseSAAstd}
			& {\rm min}     &  \frac{1}{M}\,\sum_{i=1}^M \, \sum_{k=1}^n \,
			 \left[  \frac{1}{2} \,(\efmap^k)^2(\sti) - \log  \partial_k \efmap^k(\sti) 
			 \right]   \\
			& {\rm s.t.}      &    
			\partial_k \efmap^k > 0 \qquad k=1,\ldots,n\constrtarae      \nonumber  \\
	        &  				    & \efmap \in \spaceMapT       \nonumber 
\end{eqnarray}
where we use the identity 
$\log \det \defmap \equiv \sum_{k=1}^n \, \log  \partial_k \efmap^k$,
which holds for triangular maps.
Almost magically, the objective function and the constraining set
of \eqref{OptimInverseSAAstd} are \textit{separable}: the $k$th component
of the inverse transport can be computed as the solution of a single convex 
optimization problem,
\begin{eqnarray}  \label{OptimInverseSAAstdComp}
			& {\rm min}     &  \frac{1}{M}\,\sum_{i=1}^M \,
			 \frac{1}{2} \,(\efmap^k)^2(\sti) - \log  \partial_k \efmap^k(\sti) 
			 \\
			& {\rm s.t.}      &    
			\partial_k \efmap^k > 0\constrtarae      \nonumber  \\
	        &  				    & \efmap^k \in \spaceMap_k       \nonumber 
\end{eqnarray}
where $\spaceMap_k$ denotes the space of smooth real-valued functions of $k$ variables.
Most importantly, \eqref{OptimInverseSAAstdComp} depends only on the $k$th
component of the map.
Thus, all the components of the inverse transport can be computed independently
and in parallel by solving optimization problems of the form \eqref{OptimInverseSAAstdComp}.
This is another major difference from the density-based construction of the 
map, where all the components of the transport must be computed 
simultaneously as a solution of \eqref{OptimDirect} (unless, for instance,
the target density can be written as the product of its marginals).

As in the previous section, the numerical solution of
\eqref{OptimInverseSAAstdComp} requires replacing the
infinite-dimensional function space $\spaceMap_k$ with a finite
dimensional subspace $\spaceMap_k^h \subset \spaceMap_k$. The
monotonicity constraint $\partial_k \efmap^k > 0$ can be discretized
and enforced at only finitely many points in the parameter space, for
instance at the samples $(\sti)_{i=1}^M$ (see Section \ref{sec:MapParam}).
However, a better approach is to use monotone parameterizations of the
triangular transport in order to turn \eqref{OptimInverseSAAstdComp}
into an unconstrained optimization problem (see Section
\ref{sec:MapParam}).  In either case, the solution of
\eqref{OptimInverseSAAstdComp} over a finite-dimensional approximation
space $\spaceMap_k^h$ yields a component of the approximate inverse
transport.  The quality of this approximation is a function of at
least two parameters of the problem: the structure of the space
$\spaceMap_k^h$ and the number of target samples $M$.  While enriching
the space $\spaceMap_k^h$ is often a straightforward task, increasing the number
of target samples can be nontrivial, especially when exact sampling
from the target density is impossible.
This is an important difference from density-based
construction of the map, wherein the objective function of
\eqref{OptimDirect} only requires integration with respect to the
reference measure and thus can be approximated to any desired degree
of accuracy during each stage of the computation.  Of course, in many
cases of interest, exact sampling from the target distribution is
possible. Consider, for instance, the joint density of data and
parameters in a typical Bayesian inference problem
\cite{Parno2014thesis} or the forecast distribution in a filtering
problem where the forward model is a stochastic difference equation.

Quantifying the quality of an approximate inverse transport is an important issue.
If the target density can be evaluated up to a normalizing constant, then the K--L
divergence between the pushforward of the target through the map and the reference density
can be estimated as
\begin{equation} \label{eq:convergenceInverse}
	 \Dkl (\, \efmap_\sharp \ntd \,||\, \rd \,) \approx \frac{1}{2} \,
	 \mathbb{V}\mathrm{ar}_{\ntd}[\, \log \td - \log \efmap\pull \rd \, ]
\end{equation}	 
up to second order terms in the limit of
$\mathbb{V}\mathrm{ar}_{\ntd}[\, \log \td - \log \efmap\pull \rd ]\rightarrow
0$.
This expression is analogous to \eqref{varianceConvergence} (see
Section \ref{sec:DensConst} for further details on this topic).  When
the target density cannot be evaluated, however, one can rely on
statistical tests to monitor convergence to the exact inverse
transport.  For instance, if the reference density is a standard
Gaussian, then we know that pushing forward the target samples
$(\sti)_i$ through the inverse transport should yield jointly Gaussian
samples, with independent and standard normal components. If the inverse
transport is only approximate, then the pushforward samples will not
have independent and standard normal components, and one can quantify
their deviation from such normality using standard statistical tests
\cite{thode2002testing}.

\subsection{Computing the inverse map}\label{sec:SampMap:inverse}
Up to now, we have shown how to compute the triangular inverse
transport $\efmap$ via convex optimization given samples from the
target density. In many problems of interest, however, the goal is to
evaluate the direct transport $\emap$, i.e., a map that pushes forward
the reference to the target measure.  Clearly, the following
relationship between the direct and inverse transports holds:
\begin{equation} \label{eq:relationDirInv}
		\emap(\rdv)=\efmap^{-1}(\rdv), \qquad \forall \rdv\in \real^n.	
\end{equation}
Thus, if we want to evaluate the direct transport at a particular 
$\rdva \in \real^n$, 
i.e., $\tdva \coloneqq \emap(\rdva)$, then by \eqref{eq:relationDirInv} we can simply invert
$\efmap$ at $\rdva$ to obtain $\tdva$. 
In particular, if $\rdva=(\rdvsa_1,\ldots,\rdvsa_n)$ and $\tdva=(\tdvsa_1,\ldots,\tdvsa_n)$,
then $\tdva$ is a solution of the following lower triangular system of equations:
\begin{equation}
 \efmap(\tdva)	= 
    \left[ 
 	\begin{array}{l}
 	\efmap^1(\tdvsa_1)\\
 	\efmap^2(\tdvsa_1,\tdvsa_2)\\
 	\vdots\\
 	\efmap^n(\tdvsa_1,\tdvsa_2,\ldots,\tdvsa_n)
 	\end{array}
 	\right]	
 	=
    \left[ 
 	\begin{array}{l}
 	\rdvsa_1\\
 	\rdvsa_2\\
 	\vdots\\
 	\rdvsa_n
 	\end{array}
 	\right]	
 	= 	
 	\rdva
\end{equation}
where the $k$th component of $S$ is just a function of the first $k$ input
variables. This system is in general nonlinear, but we can devise a simple recursion in $k$ to compute each component of $\tdva$ as
\begin{equation} \label{eq:recursionRoot}
	\tdvsa_k \coloneqq (S^k_{\tdvsa_1,\ldots,\tdvsa_{k-1}})^{-1}(\rdvsa_k), \qquad k=1,\ldots,n,
\end{equation}
where $S^k_{\tdvsa_1,\ldots,\tdvsa_{k-1}}:\real \rightarrow \real$
 is a one-dimensional function defined as
 $w \mapsto S^k( \tdvsa_1,\ldots,\tdvsa_{k-1}, w)$.
 That is, $S^k_{\tdvsa_1,\ldots,\tdvsa_{k-1}}$ is the restriction of the 
 $k$th component of the inverse transport obtained by fixing the first $k-1$
 input variables $\tdvsa_1,\ldots,\tdvsa_{k-1}$.
 Thus, $\tdva$ can be computed recursively via a sequence of $n$
 one-dimensional root-finding problems. Monotonicity of the triangular
 maps guarantees that \eqref{eq:recursionRoot} has a unique
 real solution for each $k$ and any given $\rdva$.   
 Here, one can use any off-the-shelf root-finding
 algorithm \footnote{Roots can be found using, for instance, Newton's
   method. When
   a component of the inverse transport is parameterized using
   polynomials, however, then a more robust root-finding approach is to use a
   bisection method based on Sturm sequences (e.g.,
   \cite{Parno2014thesis}).}.
Whenever the transport is high-dimensional (e.g., hundreds or thousands of components), 
 this recursive approach might become inaccurate, as it is sequential in nature. 
 In this case, we recommend running a few Newton iterations of the form
 \begin{equation}
 	\tdv_{j+1} = \tdv_j - \defmap(\tdv_j)^{-1}(S(\tdv_j)-\rdva)	
 \end{equation}
 to clean up the approximation of the root $\tdva$ obtained from 
 the recursive algorithm \eqref{eq:recursionRoot}.
 
 An alternative way to evaluate the direct transport is to build a parametric
 representation of $\emap$ itself via standard regression techniques.
 In particular, if $\{\tdv_1 ,\ldots,\tdv_M \}$ are samples from the target distribution, then 
 $\{\rdv_1 ,\ldots,\rdv_M \}$, with $\rdv_k \coloneqq S(\tdv_k)$ for $k=1,\ldots,M$, are
 samples from the reference distribution.
 Note that there is a one-to-one correspondence between target  and reference samples.
 Thus, we can use these pairs of samples to define a simple
 constrained least-squares problem to approximate the direct transport
 as:
\begin{eqnarray}  
			& {\rm min}     &  \sum_{k=1}^{M} \sum_{i=1}^{\pd} \,
 	\left(\emap^i( \rdv_k ) - \tdv_k \right)^2   \\
			& {\rm s.t.}      &    
			\partial_i \emap^i > 0 \qquad i=1,\ldots,n\constrrefae      \nonumber  \\
	        &  				    & \emap \in \spaceMapT       \nonumber \label{eq:lsqmap}.
\end{eqnarray} 
In particular, each component of the direct transport can be approximated independently
(and in parallel) as the minimizer of 
\begin{eqnarray}  
			& {\rm min}     &  \sum_{k=1}^{M} 
 	\left(\emap^i( \rdv_k ) - \tdv_k \right)^2  \\
			& {\rm s.t.}      &    
			\partial_i \emap^i > 0 \constrrefae      \nonumber  \\
	        &  				    & \emap^i \in  \spaceMap_i      \nonumber \label{eq:lsqmapcomp},
\end{eqnarray}  
where $ \spaceMap_i$ denotes the space of smooth real-valued functions of $i$ variables.
Of course, the numerical solution of \eqref{eq:lsqmapcomp} requires the suitable choice of 
a finite dimensional approximation space $\spaceMap_i^h \subset \spaceMap_i$.  
\medskip

\section{Parameterization of transport maps}
\label{sec:MapParam}

As noted in the previous sections, the optimization problems that one solves to obtain either the direct or inverse transport must, at some point, introduce discretization. In particular, we must define finite-dimensional approximation spaces (e.g., $\spaceMapT^h$) within which we search for a best map. In this section we describe several useful choices for $\spaceMapT^h$ and the associated map parameterizations. Closely related to the map parameterization is the question of how to enforce the monotonicity constraints $\partial_k T^k > 0$ or $\partial_k S^k > 0$ over the support of the reference and target densities, respectively. For some parameterizations, we will explicitly introduce discretizations of these monotonicity constraints. A different map parameterization, discussed in Section~\ref{s:mapMonotone}, will satisfy these monotonicity conditions automatically.

For simplicity, we will present the parameterizations below mostly in the context of the direct transport $T$. But these parameterizations can be used interchangeably for both the direct and inverse transports.

\subsection{Polynomial representations}\label{sec:MapParam:poly}
A natural way to parameterize each component of the map $T$ is by expanding it in a basis of multivariate polynomials.  We define each multivariate polynomial  $\psi_{\mathbf{j}}$ as a product of $n$ univariate polynomials, specified via a multi-index $\mathbf{j} = (j_1,j_2,\ldots ,j_n)\in \mathbb{N}_0^n$, as:
\begin{equation}
\psi_{\mathbf{j}}(\rdv) = \prod_{i=1}^n \varphi_{j_i}(\rdvs_i),\label{eq:multiPolyForm}
\end{equation}
where $\varphi_{j_i}$ is a univariate polynomial of degree $j_i$. The univariate polynomials can be chosen from any standard orthogonal polynomial family (e.g., Hermite, Legendre, Laguerre) or they can even be monomials. %
That said, it is common practice in uncertainty quantification to choose univariate polynomials that are orthogonal with respect to the input measure, which in the case of the direct transport is $\rmeas$. If $\rmeas$ is a standard Gaussian, the $(\varphi_i)_i$ above would be (suitably scaled and normalized) Hermite polynomials. The resulting map can then be viewed as a \textit{polynomial chaos} expansion \cite{Xiu2002,LeMaitre2010} of a random variable distributed according to the target measure. From the coefficients of this polynomial expansion, moments of the target measure can be directly---that is, \textit{analytically}---evaluated.
In the case of inverse transports, however,  $\tmeas$ is typically not among the canonical distributions found in the Askey scheme for which standard orthogonal polynomials can be easily evaluated. While it is possible to construct orthogonal polynomials for more general measures \cite{Gautschi1996}, the relative benefits of doing so are limited, and hence with inverse transports we do not employ a basis orthogonal with respect to $\tmeas$.

Using multivariate polynomials given in \eqref{eq:multiPolyForm}, we can express each component of the transport map $T \in \spaceMapT^h$ as
\begin{equation}
T^k(\rdv) = \sum_{\mathbf{j}\in\mathcal{J}_k} \mapp_{k,\mathbf{j}}  \, \psi_{\mathbf{j}}(\rdv), \ \ k=1, \ldots, n \label{eq:polyexpand},
\end{equation}
where $\mathcal{J}_k$ is a set of multi-indices defining the polynomial terms in the expansion for dimension $k$ and $\mapp_{k,\mathbf{j}} \in \mathbb{R}$ is a scalar coefficient. Importantly, the proper choice of each multi-index set $\mathcal{J}_k$ will force $T$ to be lower triangular. 
For instance, a standard choice of  $\mathcal{J}_k$ involves restricting each map component to a total-degree polynomial space:
\begin{equation}
\label{totaldegreepolyspace}
\mathcal{J}^{TO}_{k} = \{ \mathbf{j} : \Vert  \mathbf{j}\Vert_1\leq p \, \wedge \,   j_i=0, \ \forall i>k\}, \ k=1,\ldots, n
\end{equation}
The first constraint in this set, $ \Vert  \mathbf{j} \Vert_1\leq p$, limits the total degree of each polynomial to $p$ while the second constraint, $j_i=0, \  \forall i>k$, forces $T$ to be lower triangular. Expansions built using $\mathcal{J}^{TO}_{d}$ are quite ``expressive'' in the sense of being able to capture complex nonlinear dependencies in the target measure.  However, the number of terms in $\mathcal{J}_{k}^{TO}$ grows rapidly with $k$ and $p$.  A  smaller multi-index set can be obtained by removing all the mixed terms in the basis:
\[
  \mathcal{J}^{NM}_{k} = \{ \mathbf{j} : \Vert\mathbf{j}\Vert_1 \leq p \, \wedge \,   j_i \,  j_\ell=0, \ \forall i \neq \ell \, \wedge \,    j_i=0, \ \forall  i > k \}.
\] 
An even more parsimonious option is to use diagonal maps, via the multi-index sets
\[
  \mathcal{J}^{D}_{k} = \{ \mathbf{j} : \Vert\mathbf{j}\Vert_1 \leq p \, \wedge \,  j_i=0, \  \forall i \neq  k\}.
\] 
Figure \ref{fig:multindexSets} illustrates the difference between these three sets for $p=3$ and $k = 2$. 

An alternative to using these standard and isotropic bases is to adapt the polynomial approximation space to the problem at hand. This becomes particularly important in high dimensions. For instance, beginning with linear maps (e.g., $\mathcal{J}^{TO}_{k}$ with $p=1$) \cite{Moselhy2011} introduces an iterative scheme for enriching the polynomial basis, incrementing the degree of $\spaceMapT^h$ in a few input variables at a time. Doing so enables the construction of a transport map in $O(100)$ dimensions. In a different context, \cite{Parno2015} uses the conditional independence structure of the posterior distribution in a multiscale inference problem to enable map construction in $O(1000)$ dimensions. Further comments on adapting $\spaceMapT^h$ are given in Section~\ref{sec:Conc}.

\begin{figure}[h]\label{fig:multindexSets}
\centering 
\begin{tikzpicture}

\node[anchor=south] at (2,4.6) {Terms in multi-index set};

\foreach \i in {-0.5,0.5,...,4.5} {
        \draw [very thin,gray] (\i,-0.5) -- (\i,4.5);
    }
    \foreach \i in {0,1,...,4} {
        \node[anchor=north] at (\i,-0.55) {$\i$};
    }
    \foreach \i in {-0.5,0.5,...,4.5} {
        \draw [very thin,gray] (-0.5,\i) -- (4.5,\i);
    }
   \foreach \i in {0,1,...,4} {
        \node[anchor=east] at (-0.55,\i) {$\i$};
    }
    
    \node[anchor=north] at (2,-1) {$j_1$};
  \node[anchor=east] at (-1,2) {$j_2$};
  
      \foreach \i in {0,...,3}{
        \foreach \j in {0,...,\i}{
          \draw[fill=green!70!black] (3-\i,\j) circle (0.45);
        }
      }
    \foreach \i in {0,...,3}{
        \draw[fill=blue] (0,3-\i) circle (0.3);
      }
    \foreach \i in {0,...,3}{
        \draw[fill=red!70!black] (3-\i,0) circle (0.15);
        \draw[fill=red!70!black] (0,3-\i) circle (0.15);
      }

\node[anchor=west] at (5,3) {\tikz{ \draw[fill=green!70!black] circle (0.7ex); } Total-degree set $\mathcal{J}^{TO}_{2}$};
\node[anchor=west] at (5,2) {\tikz{ \draw[fill=blue] circle (0.7ex); } Set of diagonal terms $\mathcal{J}^{D}_{2}$};
\node[anchor=west] at (5,1) {\tikz{ \draw[fill=red!70!black] circle (0.7ex); } Set with no mixed terms $\mathcal{J}^{NM}_{2}$};
\end{tikzpicture}

\caption[Visualization of multi-index sets.]{Visualization of multi-index sets for the second component of a two dimensional map, $T^2(x_1,x_2)$.  In this case, $j_1$ is the degree of a basis polynomial in $x_1$ and $j_2$ is the degree in $x_2$.  A filled circle indicates that a term is present in the set of multi-indices.}
\end{figure}
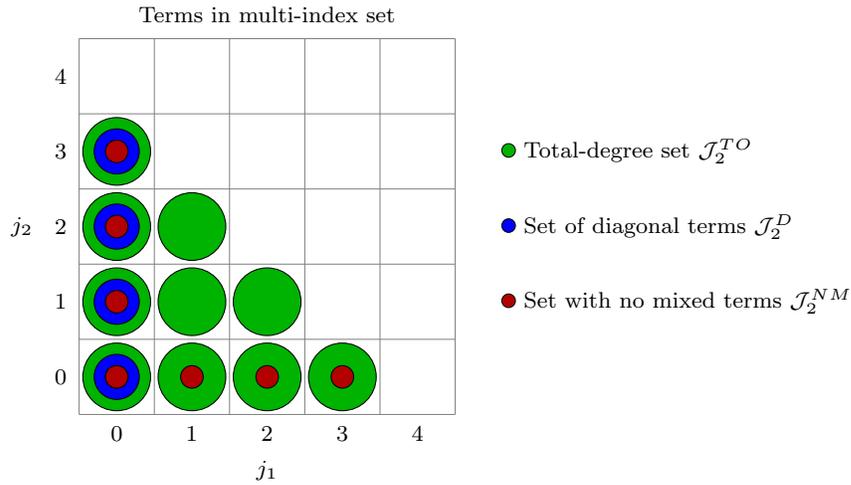

\subsection{Radial basis functions}\label{sec:MapParam:rbf}

An alternative to a polynomial parameterization of the map is to employ a combination of linear terms and radial basis functions. This representation can be more efficient than a polynomial representation in certain cases---for example, when the target density is multi-modal. The general form of the expansion in \eqref{eq:polyexpand} remains the same, but we replace polynomials of degree greater than one with radial basis functions as follows:
\begin{equation}
T^k(\rdv) = a_{k,0} + \sum_{j=1}^k a_{k,j} x_j + \sum_{j=1}^{P_k}b_{k,j}\phi_j(x_1,x_2,\ldots ,x_k; \bar{\rdv}^{k,j}), \ k=1, \ldots, n\label{eq:rbfExpand},
\end{equation}
where $P_k$ is the total number of radial basis functions used for the $k$th component of the map and $\phi_j(x_1,x_2,\ldots ,x_k; \bar{\rdv}^{k,j})$ is a radial basis function centered at 
$\bar{\rdv}^{k,j}\in\real^k$. Note that this representation ensures that the overall map $T$ is lower triangular. The $a$ and $b$ coefficients can then be exposed to the optimization algorithm used to search for the map.

Choosing the centers and scales of the radial basis functions can be challenging in high dimensions, though some heuristics for doing so are given in \cite{Parno2014thesis}. To circumvent this difficulty, \cite{Parno2014thesis} also proposes using only univariate radial basis functions and embedding them within a composition of maps.

\subsection{Monotonicity constraints and monotone parameterizations}
\label{s:mapMonotone}
Neither the polynomial representation \eqref{eq:polyexpand} nor the radial basis function representation \eqref{eq:rbfExpand} yield monotone maps for all values of the coefficients. With either of these choices for the approximation space $\spaceMapT^h$, we need to enforce the monotonicity constraints explicitly. Recall that, for the triangular maps considered here, the monotonicity constraint reduces to requiring that $\partial_k T^k > 0$ over the entire support of the reference density, for $k=1, \ldots, n$. It is difficult to enforce this condition everywhere, so instead we choose a finite sequence of points $(\sri)_i$---a stream of samples from the reference distribution, and very often the same samples used for the sample-average approximation of the objective \eqref{OptimDirectApprox}---and enforce local monotonicity at each point: $\partial_k T^k(\sri) > 0$, for $k=1, \ldots, n$. The result is a finite set of linear constraints. Collectively these constraints are weaker than requiring monotonicity everywhere, but as the cardinality of the sequence $(\sri)_i$ grows, we have stronger guarantees on the monotonicity of the transport map over the entire support of the reference density. When monotonicity is lost, it is typically only in the tails of $\rmeas$ where samples are fewer. We should also point out that the $(-\log \det \demap)$ term in the objective of \eqref{OptimDirect} acts as a barrier function for the constraint $\demap \succ 0$ \cite{renegar2001mathematical}. 

A more elegant alternative to discretizing and explicitly enforcing the monotonicity constraints is to employ parameterizations of the map that are in fact \textit{guaranteed} to be monotone
\cite{bigoni2016monotone}. Here we take advantage of the fact that monotonicity of a triangular function can be expressed in terms of one-dimensional monotonicity of its components. A smooth monotone increasing function of one variable, e.g., the first component of the lower triangular map, can be written as 
\cite{ramsay1998estimating}:
\begin{equation}
  T^1(x_1)= a_1+ \int_0^{x_1} \exp \left ( b_1( w ) \right )\, {\rm d} w ,
\end{equation}
where $a_1 \in \mathbb{R}$ is a constant and $b_1: \mathbb{R} \to \mathbb{R}$ is an arbitrary function. This can be generalized to the $k$th component of the map as:
\begin{equation}
  T^k(x_1 , \ldots , x_k)= a_k(x_1,\ldots,x_{k-1}) + \int_0^{x_k} \exp \left ( b_k(x_1,\ldots,x_{k-1}, w) \right )\, {\rm d} w
\end{equation}
for some functions $a_k: \mathbb{R}^{k-1} \to \mathbb{R}$ and $b_k: \mathbb{R}^{k} \to \mathbb{R}$. Note that $a_k$ is not a function of the $k$th input variable. Of course, we now have to pick a finite-dimensional parameterization of the functions $a_k,b_k$, but the monotonicity constraint is automatically enforced since
	\begin{equation}
	 \partial_k T^k(\rdv) = \exp \left ( b_k( x_1 , \ldots , x_k)  \right ) > 0, 
   \qquad \forall \rdv \in \real^n,
	\end{equation}
and for all choices of $a_k$ and $b_k$. Enforcing monotonocity through the map parameterization in this way, i.e., choosing $\spaceMapT^h$ so that it only contains monotone lower triangular functions, allows the resulting finite-dimensional optimization problem to be unconstrained.

\section{Related work}
\label{sec:relatedwork}

The idea of using nonlinear transformations to accelerate or simplify sampling has appeared in many different settings. Here we review several relevant instantiations.

Perhaps the closest analogue of the density-based map construction of Section~\ref{sec:DensConst} is the implicit sampling approach of 
\cite{chorin2009implicit,chorin2010implicit}. 
While implicit sampling was first proposed in the context of Bayesian 
filtering \cite{chorin2009implicit,chorin2010implicit,morzfeld2011implicit,atkins2012implicit}, it is in fact a more general scheme for importance simulation \cite{morzfeld2015parameter}. Consider the K--L divergence objective of the optimization problem \eqref{abstractOptimDirectRosenblatt}. At optimality, the K--L divergence is zero. Re-arranging this condition and explicitly writing out the arguments yields:
\begin{equation}
\Ex_{\rd} \left  [ \log \td \left ( T(\rdv) \right) +  \log \det \demap(\rdv)  - \log \beta - \log \rd(\rdv)  \right ] = 0,
\label{eq:mapvsIS}
\end{equation}
where $\beta$ is the normalizing constant of the unnormalized target density $\td$ \eqref{eq:normconst}. Now let $\tdv = T(\rdv)$. 
The central equation in implicit sampling methods is \cite{chorin2009implicit}:
\begin{equation}
\log \td (\tdv)  - \mathcal{C} = \log \rd(\rdv),
\label{eq:implicitsample}
\end{equation}
where $\mathcal{C}$ is an easily computed constant. Implicit sampling first draws a sample $\sri$ from the reference density $\eta$ and then seeks a corresponding $\sti$ that satisfies \eqref{eq:implicitsample}. This problem is generally underdetermined, as terms in \eqref{eq:implicitsample} are scalar-valued while the samples $\sri, \sti$ are in $\mathbb{R}^n$. Accordingly, the random map implementation of implicit sampling 
\cite{Morzfeld2012} restricts the search for  $\sti$ to a one-dimensional optimization problem along randomly-oriented rays emanating from a point in $\mathbb{R}^n$, e.g., the mode of the target distribution. This scheme is efficient to implement, though it is restricted to target densities whose contours are star-convex with respect to the chosen point 
\cite{goodman2015small}. 
Satisfying \eqref{eq:implicitsample} in this way defines the \textit{action} of a map from $\rd$ to another distribution, and the intent of implicit sampling is that this pushforward distribution should be close to the target. There are several interesting contrasts between \eqref{eq:mapvsIS} and \eqref{eq:implicitsample}, however. First is the absence of the Jacobian determinant in \eqref{eq:implicitsample}. The samples $\sti$ produced by implicit sampling must then (outside of the Gaussian case) be endowed with weights, which result from the Jacobian determinant of the implicit map.
The closeness of the implicit samples to the desired target is reflected in the variation of these weights. A second contrast is that \eqref{eq:mapvsIS} is a global statement about the action of a map $T$ over the entire support of $\rd$, wherein the map $T$ appears explicitly. On the other hand, \eqref{eq:implicitsample} is a relationship between points in $\mathbb{R}^n$. The map does not appear explicitly in this relationship; rather, the way in which  \eqref{eq:implicitsample} is satisfied {\it implicitly} defines the map.

Another optimization-based sampling algorithm, similar in spirit to implicit sampling though different in construction, is the randomize-then-optimize (RTO) approach of \cite{Bardsley2014}. This scheme is well-defined for target distributions whose log-densities can be written in a particular quadratic form following a transformation of the parameters, 
with some restrictions on the target's degree of non-Gaussianity. The algorithm proceeds in three steps. First, one draws a sample $\sri$ from a Gaussian (reference) measure, and uses this sample to fix the objective of an unconstrained optimization problem in $n$ variables. Next, one solves this optimization problem to obtain a sample $\sti$. And finally, this sample is ``corrected'' either via an importance weight or a Metropolis step. The goal, once again, is that the distribution of the samples $\sti$ should be close to the true target $\ntd$---though as in implicit sampling, outside of the Gaussian case these two distributions will not be identical and the correction step is required.  

Hence, another way of understanding the contrast between these optimization-based samplers and the transport map framework is that the latter defines an optimization problem \textit{over maps}, where minimizing the left-hand side of \eqref{eq:mapvsIS} is the objective. Implicit sampling and RTO instead solve simpler optimization problems \textit{over samples}, where each minimization yields the action of a particular transport. 
A crucial feature of these transports is that the pushforward densities they induce can be evaluated in closed form, thus allowing implicit samples and RTO samples to be reweighted or Metropolized in order to obtain asymptotically unbiased estimates of target expectations.
Nonetheless, implicit sampling and RTO each implement a particular transport, and they are \textit{bound} to these choices. In other words, these transports cannot be refined, and it is difficult to predict their quality for arbitrarily non-Gaussian targets. The transport map framework instead implements a search over a space of maps, and therefore contains a tunable knob between computational effort and accuracy: by enriching the search space $\spaceMapT^h$, one can get arbitrarily close to any target measure. Of course, the major disadvantage of the transport map framework is that one must then parameterize maps $T \in \spaceMapT^h$, rather than just computing the action of a particular map. But parameterization subsequently allows direct evaluation and sampling of the pushforward $T_\sharp \rd$ without appealing again to the target density.

Focusing for a moment on the specific problem of Bayesian inference, another class of approaches related to the transport map framework are the sampling-free Bayesian updates introduced in 
\cite{rosic2012sampling,litvinenko2013uncertainty,litvinenko2013inverse,
matthies2015inverse,saad2009characterization}. 
These methods treat Bayesian inference as a projection. In particular, they approximate the conditional expectation of any prescribed function of the parameters, where conditioning is with respect to the $\sigma$-field generated by the data. The approximation of the conditional expectation may be refined by enlarging the space of functions (typically polynomials) on which one projects; hence one can generalize linear Bayesian updates \cite{rosic2012sampling} to nonlinear Bayesian updates \cite{litvinenko2013uncertainty}. 
The precise goal of these approximations is different from that of the transport map framework, however.  Both methods approximate random variables, but different ones: \cite{litvinenko2013uncertainty} focuses on the conditional expectation of a function of the parameters (e.g., mean, second moment) as a function of the data random variable, whereas the transport approach to inference \cite{Moselhy2011} aims to fully characterize the posterior random variable for a particular realization of the data. 

Ideas from \textit{optimal} transportation have also proven useful in the context of Bayesian inference. In particular, \cite{Reich2013} solves a \textit{discrete} Kantorovich optimal transport problem to find an optimal transport plan from a set of unweighted samples representing the prior distribution to a weighted set of samples at the same locations, where the weights reflect the update from prior to posterior. This transport plan is then used to construct a linear transformation of the prior ensemble that yields consistent posterior estimates. The linear transformation can be understood as a resampling strategy, replacing the weighted samples with new samples that are convex combinations of the prior samples. The ability to ``move'' the samples 
within the convex hull of the prior ensemble leads to improved performance over other resampling strategies, though the prior samples should then have good coverage of the support of the posterior.

Turning to the sample-based map construction of Section~\ref{sec:SampConst}, it is interesting to note that attempts to Gaussianize collections of samples using nonlinear transformations date back at least to 1964 \cite{Box1964}. In the geostatistics literature, the notion of Gaussian anamorphosis \cite{wackernagel2013multivariate} uses the empirical cumulative distribution function (CDF), or Hermite polynomial approximations of the CDF, to Gaussianize the \textit{marginals} of multivariate data. These transformations do not create joint Gaussianity, however. 

To construct joint transformations of dependent multivariate data, \cite{Stavropoulou2015} proposes a scheme employing discrete optimal transport. This approach generates an equivalent number of samples from a reference measure; solves a discrete assignment problem between the two sample sets, given a quadratic transport cost; and uses the resulting pairs to estimate a polynomial map using linear regression. This is a two-stage approach, in contrast with the 
single convex optimization problem proposed in Section~\ref{sec:SampConst}. For reference and target distributions with compact support, it yields an approximation of the Monge optimal transport rather than the Knothe-Rosenblatt rearrangement. 

Moving from polynomial approximations to nonparametric approaches, 
\cite{Tabak2013,laurence2014constrained,tabak2014data} introduce schemes for multivariate density estimation based on progressively transforming a given set of samples to a (joint) standard Gaussian by \textit{composing} a sequence of monotone maps. The maps are typically chosen to be rather simple in form (e.g., sigmoid-type functions of one variable). In this context, we note that the empirical K--L divergence objective in \eqref{OptimInverseSAA} is the pullback density $S^{-1}_\sharp \eta$ evaluated at the samples $(\sti)_{i=1}^M$, and hence can be viewed as the log-likelihood of the map $S$ given the target samples. 
In \cite{Tabak2013}, each new element of the composition of maps is guaranteed not to decrease this log-likelihood function. A related scheme is presented in 
\cite{Laparra2011}; 
here the sequence of maps alternates between rotations and diagonal maps that transform the marginals. Rotations are chosen via principle component analysis (PCA) or independent component analysis (ICA). The resulting composition of maps can Gaussianize remarkably complex distributions in hundreds of dimensions (e.g., samples from a face database). Both of these methods, however, reveal an interesting tension between the number of maps in the composition and the complexity of a single map. When each map in the composition is very simple (e.g., diagonal, or even constant in all but one variable) the maps are easy to construct, but their composition can converge very slowly to a Gaussianizing transformation. On the other hand, we know that there exist maps (e.g., the Knothe-Rosenblatt rearrangement or the Monge optimal transport map) that can Gaussianize the samples immediately, but approximating them directly requires much more effort. Some of these tradeoffs are explored in \cite{Parno2014thesis}.

Yet another use of transformations in stochastic simulation is the warp bridge sampling approach of 
\cite{meng2002warp}. The goal of bridge sampling is to estimate the ratio of normalizing constants of two probability densities (e.g., the ratio of Bayesian model evidences). 
\cite{meng2002warp} introduces several deterministic and/or stochastic transformations to increase the overlap of the two densities---by translating, scaling, and even symmetrizing them. These transformations can reduce the asymptotic variance of the bridge sampling estimator. More recent generalizations use Gaussian mixture approximations of the densities to design transformations suitable for multi-modal problems \cite{wang2015methods}.

Finally, setting aside the notion of nonlinear transformations, it is useful to think of the minimization problem \eqref{abstractOptimDirectRosenblatt} in the broader context of variational Bayesian methods 
\cite{attias1999inferring,
jaakkola2000bayesian,wainwright2008graphical,fox2012tutorial}. 
As in typical variational Bayesian approaches, we seek to approximate some complex or intractable distribution (represented by $\ntd$) with a simpler one. But the approximating distribution in the transport map framework is any pushforward of a reference density. In contrast with variational Bayesian approaches, this distribution can be found %
without imposing strong assumptions on its factorization (e.g., the mean field approximation) or on the family of distributions from which it is drawn (e.g., an exponential family). The transport map framework is also distinguished from variational Bayesian approaches due to the availability of the pullback density \eqref{eq:pullbackdensity}---in an intuitive sense, the ``leftover'' after approximation with any given map. Using evaluations of the pullback density, one can compose sequences of maps, enrich the approximation space of any given map, or use the current transport to precondition an exact sampling scheme. 
\section{Conditional sampling}
\label{sec:CondSamp}
In this section we will show how the triangular structure of the
transport map allows efficient sampling from particular
conditionals %
of the target density.
This capability is important because, in general, the ability
  to sample from a distribution does not necessarily provide efficient
  techniques for also sampling its conditionals. %
As in the previous sections, assume that the reference
and target measures are absolutely continuous with respect to the Lebesgue measure
on $\real^\dimtot$ with smooth and positive densities.
Let $\emap:\real^\dimtot \rightarrow \real^\dimtot$ be a triangular and monotone increasing transport  that pushes forward
the reference to the target density, i.e., $\emap_\sharp \rd = \ntd$, where $\emap$ is the 
Knothe-Rosenblatt rearrangement.

We first need to introduce some additional notation. 
There is no loss of generality in thinking of  
the target as the joint distribution of some
random vector $\rvtot$ in $\real^\dimtot$.
Consider a partition of this random vector as
$\rvtot=(\frv,\srv)$ where $\frv\in\real^{\dimfrv}$,  
$\srv\in\real^{\dimsrv}$, and $n=\dimfrv+\dimsrv$.
In other words, $\frv$ simply comprises the first $\dimfrv$ components of $\rvtot$.
We equivalently denote the joint density of $\rvtot=(\frv,\srv)$ by either $\ntd$
or $\jointfs$. That is, $\ntd \equiv \jointfs$.
We define the conditional density of $\srv$ given $\frv$ as
\begin{equation}
	\condsgf \coloneqq \frac{\jointfs}{\margf},
\end{equation}
where $\margf \coloneqq \int \jointfs(\cdot,\sdv) \,{\rm d}\sdv$ is the marginal density of $\frv$.
In particular, $\condsgf(\sdv\vert\fdv)$ is the conditional density of $\srv$ at 
$\sdv\in\real^{\dimsrv}$ given the event $\{ \frv = \fdv \}$.
Finally, we define $\condsgfd$ as a map from $\real^{\dimsrv}$ to $\real$ such that
\begin{equation}
	 \sdv \mapsto \condsgf(\sdv\vert\fdv).
\end{equation}
We can think of $\condsgfd$ as a particular normalized slice of the joint density $\jointfs$  for
$\frv=\fdv$, as shown in Figure \ref{fig:condSlice}.
Our goal is to show how the triangular transport map $\emap$  can be used to
efficiently sample the conditional density $\condsgfd$.

If $\emap$ is a monotone increasing lower triangular transport on $\real^\dimtot$, we can denote its components by
\begin{equation} \label{eq:formTransport}
 \emap(\rdv)	\coloneqq \emap(\rdvf,\rdvse) = 
    \left[ 
 	\begin{array}{l}
 	\mapcf(\rdvf)\\
 	\mapcs(\rdvf,\rdvse)
 	\end{array}
 	\right], \qquad \forall	 \rdv \in\real^{\dimtot},
\end{equation}	
where $\mapcf:\real^{\dimfrv}\rightarrow\real^{\dimfrv}$, 
$\mapcs:\real^{\dimfrv}\times\real^{\dimsrv}\rightarrow\real^{\dimsrv}$, 
and where $\rdv\coloneqq(\rdvf,\rdvse)$ is a partition of the dummy variable 
$\rdv\in\real^{\dimtot}$ as $\rdvf\in\real^{\dimfrv}$
and
$\rdvse\in\real^{\dimsrv}$, 
i.e., $\rdvf$ consists of the first $\dimfrv$ components of $\rdv$.

In the context of Bayesian inference, $\srv$ could represent the 
inversion parameters or latent variables and $\frv$ the observational
data.  In this interpretation, $\condsgfd$ is just the posterior
distribution of the parameters for a particular realization of the
data.  Sampling this posterior distribution yields an explicit
characterization of the Bayesian solution and is thus of crucial
importance. 
This scenario is particularly relevant in the context
of online Bayesian inference where one is concerned with fast posterior computations for multiple
realizations of the data (e.g., \cite{Parno2014thesis,huan16maps}).
Of course, if one is only interested in $\condsgfd$ for a single
realization of the data %
then there is no need
to first approximate the joint density $\jointfs$ and subsequently perform
conditioning to sample $\condsgfd$. 
Instead, one should simply pick $\condsgfd$ as the target density and compute
the corresponding transport \cite{Moselhy2011}.
In the latter case the dimension of the transport map would be  
independent of the size of the data.
\medskip

The following lemma shows how to efficiently sample the conditional
density $\condsgfd$ given a monotone increasing triangular transport
$\emap$.  In what follows we assume that the reference density can be
written as the product of its marginals; that is,
$\rd(\rdv)=\rdfrv(\rdvf) \rdsrv(\rdvse)$ for all $\rdv=(\rdvf,\rdvse)$
in $\real^{\dimtot}$ and with marginal densities $\rdfrv$ and $\rdsrv$.
This hypothesis is not restrictive as the reference density is a
degree of freedom of the problem (e.g., $\rd$ is often a standard
normal density).

\begin{lemma} \label{lem:conditional}
For a fixed $\fdv \in\real^{\dimfv}$, define $\solsystem$ as the unique
element of  $\real^{\dimfv}$ such that $\mapcf(\solsystem)=\fdv$.
Then,  the map
$\emap_\fdv:\real^{\dimsrv} \rightarrow \real^{\dimsv}$, defined  as
\begin{equation} \label{eq:condTransp}
	\wb \mapsto \mapcs(\,\solsystem, \wb),
\end{equation}
pushes forward $\rdsrv$ 
to the desired conditional density $\condsgfd$.
\end{lemma}

\begin{proof}
First of all, notice that $\solsystem \coloneqq (\mapcf)^{-1}(\fdv)$ is well defined since 
$\mapcf:\real^{\dimfrv}\rightarrow\real^{\dimfrv}$ is a monotone increasing and invertible function
by definition of the Knothe-Rosenblatt rearrangement $\emap$. Then:
	\begin{equation}
		\emap_\fdv\pull\,\condsgfd(\dvpl)  = 	
		\condsgf( \emap_\fdv(\dvpl) \vert \fdv )\,\vert \det  \demap_\fdv(\dvpl) \vert = 
		\frac{\jointfs( \fdv, \mapcs(\,\solsystem, \dvpl) )}{\margf(\fdv)}\,
		\det \nabla_{\dvpl}\, \mapcs(\,\solsystem, \dvpl)
	\end{equation}	
Since, by definition, $\mapcf(\solsystem)=\fdv$,	 we have for all 
$\dvpl \in \real^{\dimsrv}$:
	\begin{eqnarray}
	\emap_\fdv\pull\,\condsgfd(\dvpl) & = &
	\frac{\jointfs( \mapcf(\solsystem), \mapcs(\,\solsystem, \dvpl) )}
	{\margf(\mapcf(\solsystem))}\,
	\det \nabla_{\dvpl}\, \mapcs(\,\solsystem, \dvpl) \nonumber \\
	& = &
	\frac{\jointfs( \mapcf(\solsystem), \mapcs(\,\solsystem, \dvpl) )}
	{(\mapcf)\pull \margf(\solsystem)}\,
	\det \demap^{\frv}(\solsystem) \,\det \nabla_{\dvpl}\, \mapcs(\,\solsystem, \dvpl) \nonumber \\
	&=&
	 \frac{\emap\pull \jointfs( \solsystem , \dvpl )}
	 {\rdfrv(\solsystem)} =
	 \frac{\rd(\solsystem , \dvpl)}
	  {\rdfrv(\solsystem)} =
	  \rdsrv(\dvpl)
	 \nonumber  
	\end{eqnarray}
where we used the identity $\mapcf_\sharp\,\rdfrv = \margf $	 which follows
from the definition of Knothe-Rosenblatt rearrangement (e.g., \cite{Villani2009}).
\hfill $\square$
\end{proof}
\medskip

We can interpret the content of Lemma \ref{lem:conditional} in the
context of Bayesian inference.  If we observe a particular realization
of the data, i.e., $\frv=\fdv$, then we can easily sample the
posterior distribution $\condsgfd$ as follows.  First, solve the
nonlinear triangular system $\mapcf(\solsystem)=\fdv$ to get
$\solsystem$.  Since $\mapcf$ is a lower triangular and invertible
map, one can solve this system using the techniques described in
Section \ref{sec:SampMap:inverse}.  Then, define a new map
$\emap_\fdv:\real^{\dimsrv} \rightarrow \real^{\dimsv}$ as
$\emap_\fdv(\wb)\coloneqq \mapcs(\,\solsystem, \wb )$ for all
$\wb \in \real^{\dimsrv}$, and notice that the pushforward through the
map of the marginal distribution of the reference over
the parameters, i.e., $(\emap_\fdv)_\sharp \rdsrv$, is precisely the desired posterior
distribution.

Notice that $\emap_\fdv$ is a single transformation parameterized by $\fdv\in\real^{\dimfrv}$.
Thus it is straightforward to condition on a different value of $\frv$, say 
$\frv = \tilde{\fdv}$. 
We only need to solve a new nonlinear triangular system of the form 
$\mapcf(\boldsymbol{x}^\star_{\tilde{\fdv}})=\tilde{\fdv}$ to define a transport 
$\emap_{\tilde{\fdv}}$ according to \eqref{eq:condTransp}.
Moreover, note that the particular triangular structure of the
transport map $\emap$ is essential to achieving efficient sampling
from the conditional $\condsgfd$ in the manner described by Lemma
\ref{lem:conditional}.

\begin{figure}[h]
\centering
\input{figures/SliceDensity.tex}
\caption[Obtaining conditional density from joint density.]{
({\it left}) Illustration of a two dimensional
joint density $\jointfs$ together with a particular slice at $d=0$.
({\it right}) Conditional density $\condsgf(\sdvs\vert 0)$ obtained from a 
normalized slice of the joint density at $d=0$.
}
\label{fig:condSlice}
\end{figure}

\section{Example: biochemical  oxygen demand model}
\label{sec:SampUses}
\label{sec:SampUses:bod}
Here we demonstrate some of the measure transport approaches described
in previous sections with a simple Bayesian inference problem,
involving a model of biochemical oxygen demand (BOD) commonly used in
water quality monitoring \cite{Sullivan2010}.  This problem is a
popular and interesting test case for sampling methods (e.g., MCMC \cite{Parno2015mcmc},
RTO \cite{Bardsley2014}).  The time-dependent forward
model is defined by
\begin{equation}
\mathfrak{B}(t) = A\,(1-\exp(-B \,t))+\mathcal{E},
\label{eq:forwardmodel}
\end{equation}
where $A$ and $B$ are unknown scalar parameters modeled as random
variables, $t$ represents time, and $\mathcal{E}\sim
\Gauss(0,10^{-3})$ is an additive Gaussian observational noise that
is statistically independent of $A$ and $B$.  
In this example, the data 
$\data$ consist of five observations of $\mathfrak{B}(t)$ at $t=\{1,2,3,4,5\}$ and is
thus a vector-valued random variable defined by
 \[\data \coloneqq [\,\mathfrak{B}(1),\, \mathfrak{B}(2), \,\mathfrak{B}(3), \, 
\mathfrak{B}(4), \,\mathfrak{B}(5)\,].\] 
Our goal is to characterize the joint distribution of $A$ and $B$ conditioned on the 
observed data. 
We assume that $A$ and $B$ are independent under the prior measure, with uniformly 
distributed marginals:
\begin{equation}
A \sim \Unif(0.4,1.2), \qquad B \sim \Unif(0.01, 0.31).
\end{equation}
Instead of inferring $A$ and $B$ directly, we choose to invert for some new target 
parameters, $\param_1$ and $\param_2$, that are related to the original parameters through the CDF of a standard normal distribution:
\begin{eqnarray}
A & \coloneqq &
                \left[0.4+0.4\left(1+\erf{\left(\frac{\param_1}{\sqrt{2}}\right)}\right)\right] \label{eq:priortransf}
  \\
B & \coloneqq & \left[0.01+0.15\left(1+\erf{\left(\frac{\param_2}{\sqrt{2}}\right)}\right)\right].
\end{eqnarray}
Notice that these transformations are invertible.
Moreover, the resulting prior marginal distributions over the target
parameters $\param_1$ and $\param_2$ are given by
\begin{equation}
	\param_1\sim \Gauss(0,1), \qquad \param_2\sim \Gauss(0,1).
\label{eq:priorgauss}
\end{equation}
We denote the target random vector by
$\paramt\coloneqq(\param_1,\param_2)$.  The main advantage of
inferring $\paramt$ as opposed to the original parameters $A$ and $B$
directly is that the support of $\paramt$ is unbounded.  Thus, there
is no need to impose any geometric constraints on the range of the
transport map. 

\subsection{Inverse transport: map from samples}
\label{s:BODinverse}
We start with a problem of efficient online Bayesian inference---where
one is concerned with fast posterior computations for multiple
realizations of the data---using the inverse transport of Section
\ref{sec:SampConst}.  Our first goal is to characterize the joint
distribution of data and parameters, $\ntd_{\data,\paramt}$, by means
of a lower triangular transport.  As explained in Section
\ref{sec:CondSamp}, this transport will then enable efficient \textit{conditioning}
on different realizations of the data $\data$.  Posterior computations
associated with the conditional $\ntd_{\paramt \vert \data=\datai}$,
for arbitrary instances of $\datai$, will thus become computationally
trivial tasks.

Note that having defined the prior and likelihood via
\eqref{eq:priortransf}--\eqref{eq:priorgauss} and
\eqref{eq:forwardmodel}, respectively, we can generate arbitrarily
many independent ``exact'' samples from the joint target
$\ntd_{\data,\paramt}$. For the purpose of this demonstration, we will
pretend that we cannot evaluate the unnormalized target density and
that we can access the target only through these samples.  This is a
common scenario in Bayesian inference problems with intractable
likelihoods
\cite{wilkinson2011stochastic,marin2012approximate,Csillery2010}.
This sample-based setting is well suited to the computation of a
triangular inverse transport---a transport map that pushes forward the
target to a standard normal reference density---as the solution of a
convex and separable optimization problem \eqref{OptimInverseSAAstd}.  The
direct transport can then be evaluated implicitly using the techniques
described in Section \ref{sec:SampMap:inverse}.

We will solve \eqref{OptimInverseSAAstd} for a matrix of different numerical configurations. The expectation with respect to the target measure is discretized using different Monte Carlo sample sizes  ($5 \times 10^3$ versus $5\times 10^4$). The reference and target are measures on $\mathbb{R}^7$, and thus the finite-dimensional approximation space for the inverse transport, $\spaceMapT^h \subset \spaceMapT$, is taken to be a total-degree polynomial space in $\pd=7$ variables, parameterized with a Hermite basis and using range of different degrees \eqref{totaldegreepolyspace}. In particular, we will consider maps ranging from linear ($p=1$) to seventh degree ($p=7$). The result of \eqref{OptimInverseSAAstd} is an approximate inverse transport $\ifmap$. An approximation to the direct transport, $\imap$, is then obtained via standard regression techniques as explained in Section \ref{sec:SampMap:inverse}. In particular, the direct transport is sought in the same approximation space $\spaceMapT^h$ as the inverse transport.

In order to assess the accuracy of the computed transports, we will characterize the conditional density $\ntd_{\paramt \vert \data=\datai}$ (see Lemma \ref{lem:conditional} in Section \ref{sec:CondSamp}) for a value of the 
data given by 
\begin{equation*}
	\datai=[0.18,   0.32,   0.42,   0.49,   0.54].
\end{equation*}
Table \ref{tab:off:error} compares moments of the approximate conditional $\ntd_{\paramt \vert \data=\datai}$ computed via the inverse transports  to the ``true'' moments of this distribution as estimated via a ``reference''  adaptive Metropolis MCMC scheme \cite{Haario2001}. While more efficient MCMC algorithms exist, the adaptive Metropolis algorithm is well known and enables a qualitative comparison of the computational cost of the transport approach to that of a widely used and standard method. The MCMC sampler was tuned to have an acceptance rate of 26\%.  The chain was run for $6\times 10^6$ steps, $2\times 10^4$  of which were discarded as burn-in. The moments of the approximate conditional density $\ntd_{\paramt \vert \data=\datai}$ given by each computed transport are estimated using $3\times 10^4$ independent samples generated from the conditional map. The accuracy comparison in Table \ref{tab:off:error} shows that the cubic map captures the mean and variance of $\ntd_{\paramt \vert \data=\datai}$ but does not accurately capture the higher moments.  
Increasing the map degree, together with the number of target samples, yields better estimates of these moments. 
For instance, degree-seven maps constructed with $5\times 10^4$ target samples can reproduce the skewness and kurtosis of the conditional density reasonably well.
Kernel density estimates of the two-dimensional conditional density  $\ntd_{\paramt \vert \data=\datai}$ using  $50\, 000$ samples are also shown in Figure \ref{fig:approxPost50000} for different orders of the computed transport. The degree--seven map gives results that are nearly identical to the MCMC reference computation.

In a typical application of Bayesian inference, we can regard the time
required to compute an approximate inverse transport $\ifmap$ and
a corresponding approximate direct transport $\imap$ as ``offline''
time.  This is the expensive step of the computations, but it is
independent of the observed data. The ``online'' time is that
required to generate samples from the conditional distribution
$\ntd_{\paramt \vert \data=\datai}$ when a new realization of the
data $\{\data=\datai\}$ becomes available.  The online step is
computationally inexpensive since it requires, essentially, only
the solution of a single nonlinear triangular system of the dimension
of the data (see Lemma \ref{lem:conditional}).

In Table \ref{tab:off:time} we compare the computational time of the map-based approach to that of the reference adaptive Metropolis MCMC  scheme.  The online time column shows how long each method takes to generate $30\, 000$ independent samples from the conditional $\ntd_{\paramt \vert \data=\datai}$.  For MCMC, we use the average amount of time required to generate a chain with an effective sample size of $30\, 000$.
Measured in terms of online time, the polynomial transport maps are roughly two orders of magnitude more efficient than the adaptive MCMC sampler.  More sophisticated MCMC samplers could be used, of course, but the conditional map approach will retain a significant advantage because, after solving for $\solsystem$ in Lemma \ref{lem:conditional}, it can generate \textit{independent} samples at negligible cost.  We must stress, however, that the samples produced in this way are from an \textit{approximation} of the targeted conditional. In fact, the conditioning lemma holds true only if the computed joint transport is exact, and a solution of \eqref{OptimInverseSAAstd} is an approximate  transport for the reasons discussed in Section \ref{subsec:convex}. Put in another way, the conditioning lemma is exact for the pushforward of the approximate map, $\ntda = \imap_\sharp \rd$. Nevertheless, if the conditional density $\ntd_{\paramt \vert \data=\datai}$  can be evaluated up to a normalizing constant, then one can quantify the error in the approximation of the conditional using \eqref{eq:convergenceInverse}.
Under these conditions, if one is not satisfied with this error, then any of the approximate maps $\imap_\fdv$ constructed here could be useful as a proposal mechanism for importance sampling or MCMC, to generate (asymptotically) exact samples from the conditional of interest. For example, the map constructed using the offline techniques discussed in this example could provide an excellent initial map for the MCMC scheme of \cite{Parno2015mcmc}.  
Without this correction, however, the sample-based construction of lower triangular inverse maps, coupled with direct conditioning, can be seen as a flexible scheme for fast approximate Bayesian computation.

\begin{table}[h]
\centering
\caption[Accuracy of offline inference maps on BOD problem.]{
BOD problem of Section \ref{s:BODinverse} via inverse transport and conditioning. 
First four moments of the conditional density $\ntd_{\paramt \vert \data=\datai}$, for
$\datai=[0.18,   0.32,   0.42,   0.49,   0.54]$, estimated by a ``reference'' run of an adaptive Metropolis sampler with $6\times 10^6$ steps, and by transport maps up to degree $p=7$ with $3\times 10^4$ samples.}
\medskip
\footnotesize
\begin{tabular}{|c|c|rr|rr|rr|rr|}\hline
\multirow{2}{*}{map type} & \# training& \multicolumn{2}{c|}{mean} &  \multicolumn{2}{c|}{variance} &  \multicolumn{2}{c|}{skewness} &  \multicolumn{2}{c|}{kurtosis}\\
& samples & $\param_1$ & $\param_2$ & $\param_1$ & $\param_2$ & $\param_1$ & $\param_2$ & $\param_1$ & $\param_2$ \\\hline
\rowcolor{black!10} \multicolumn{2}{|c|}{MCMC ``truth''} &  0.075 & 0.875 & 0.190 & 0.397 & 1.935 & 0.681 & 8.537 & 3.437\\\hline
\multirow{2}{*}{$p=1$} & 5000 & 0.199 & 0.717 & 0.692 & 0.365 & -0.005 & 0.010 & 2.992 & 3.050\\
&  50000 & 0.204 & 0.718 & 0.669 & 0.348 & 0.016 & -0.006 & 3.019 & 3.001\\\hline
\rowcolor{blue!10}  & 5000 & 0.066 & 0.865 & 0.304 & 0.537 & 0.909 & 0.718 & 4.042 & 3.282\\
\rowcolor{blue!10} \multirow{-2}{*}{$p=3$} & 50000 & 0.040 & 0.870 & 0.293 & 0.471 & 0.830 & 0.574 & 3.813 & 3.069\\\hline
 \multirow{2}{*}{$p=5$} & 5000 & 0.027 & 0.888 & 0.200 & 0.447 & 1.428 & 0.840 & 5.662 & 3.584 \\
 & 50000 & 0.018 & 0.907 & 0.213 & 0.478 & 1.461 & 0.843 & 6.390 & 3.606 \\\hline
\rowcolor{blue!10}  & 5000 & 0.090 & 0.908 & 0.180 & 0.490 & 2.968 & 0.707 & 29.589 & 16.303\\
\rowcolor{blue!10} \multirow{-2}{*}{$p=7$} & 50000 & 0.034 & 0.902 & 0.206 & 0.457 & 1.628 & 0.872 & 7.568 & 3.876\\\hline
 \end{tabular}
\label{tab:off:error}
\end{table}

\begin{figure}[h]
\centering
\subfigure[MCMC truth]{\begin{tikzpicture}
\begin{axis}[enlargelimits=false, axis on top, ticks=none, height=\jointPlotHeight,width=\jointPlotWidth]
\addplot graphics [xmin=-0.650000,xmax=1.500000,ymin=-0.200000,ymax=2.500000] {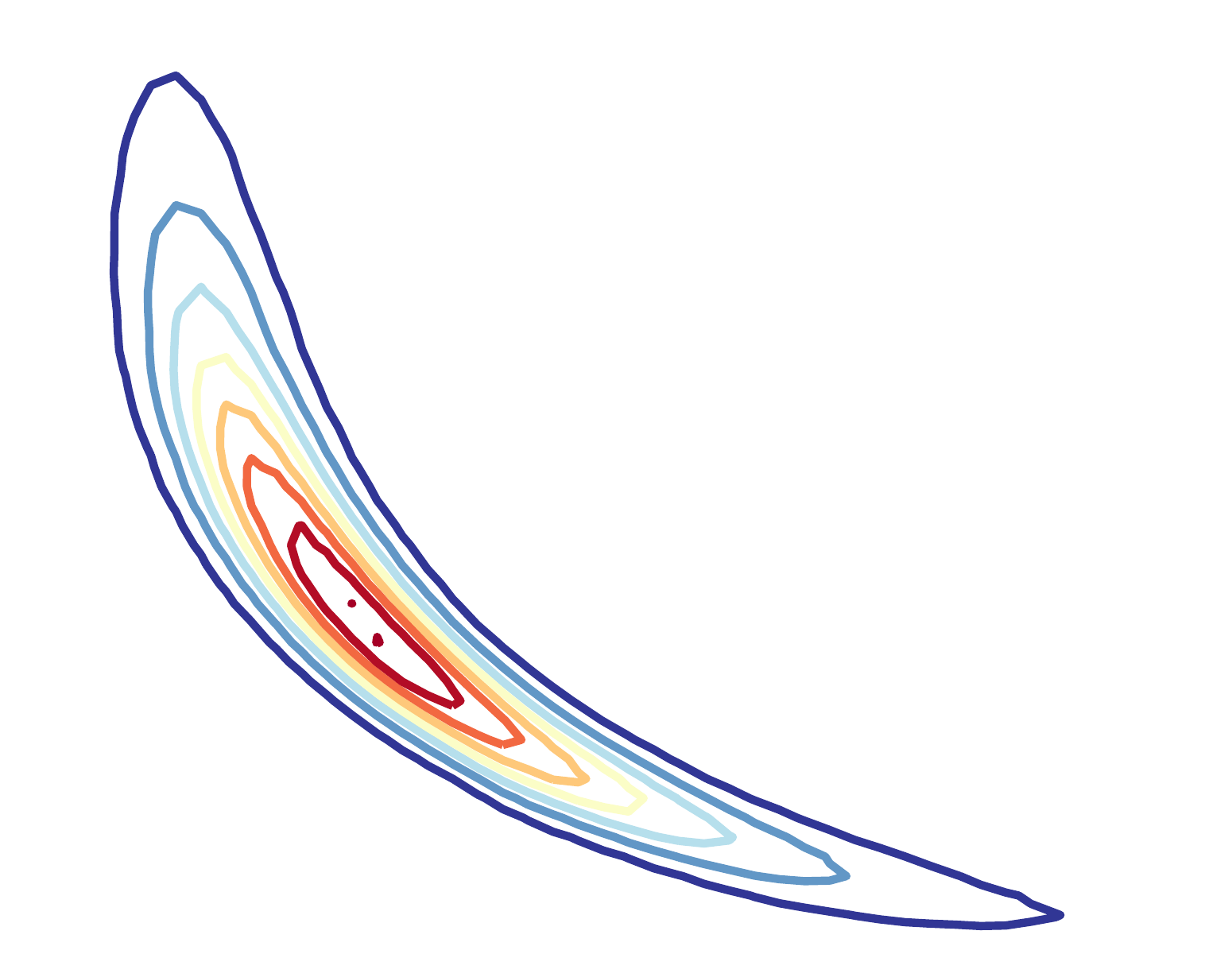};
\end{axis}
\end{tikzpicture}}\subfigure[degree $p=3$]{\begin{tikzpicture}
\begin{axis}[enlargelimits=false, axis on top, ticks=none, height=\jointPlotHeight,width=\jointPlotWidth]
\addplot graphics [xmin=-0.650000,xmax=1.500000,ymin=-0.200000,ymax=2.500000] {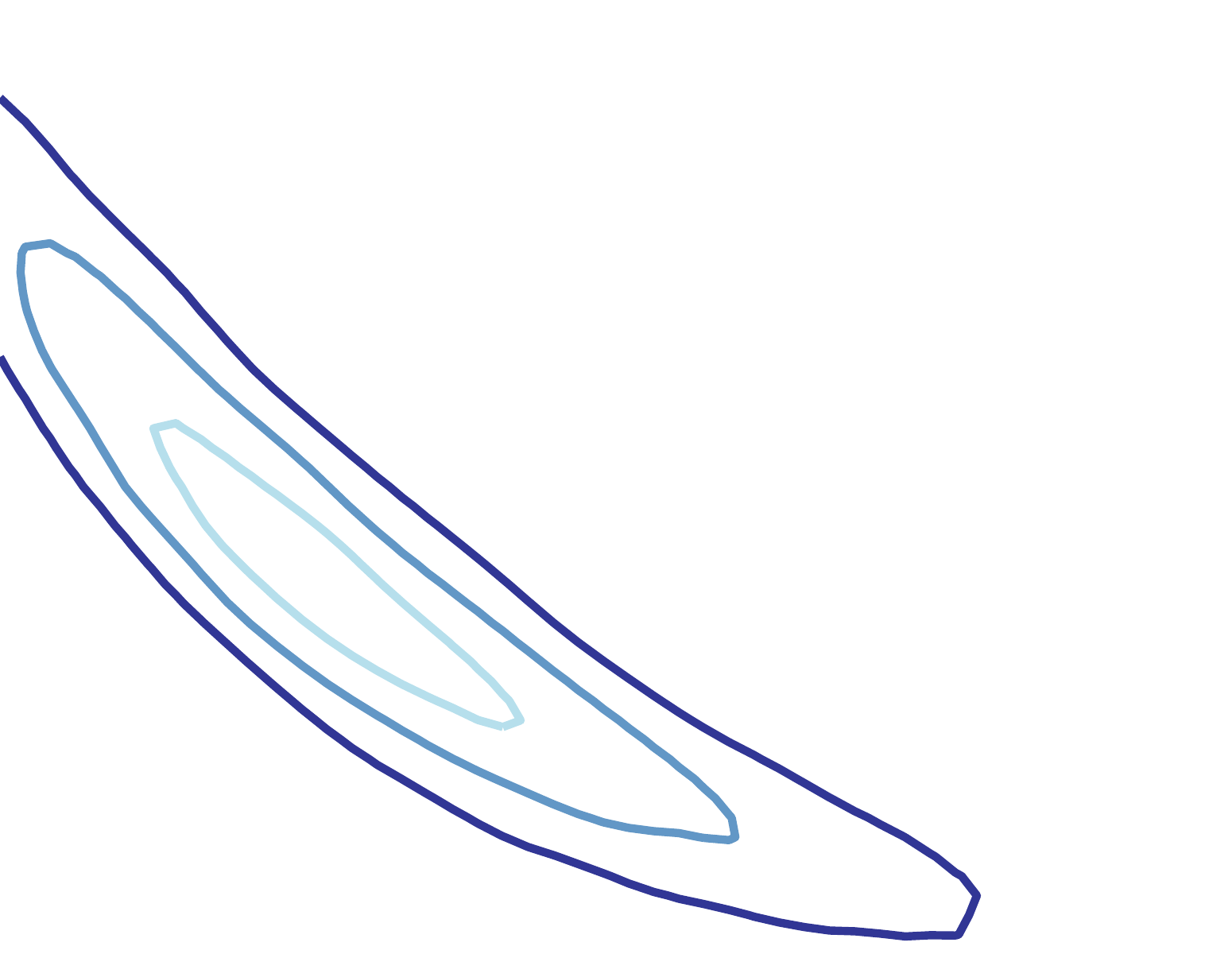};
\end{axis}
\end{tikzpicture}}
\subfigure[degree $p=5$]{\begin{tikzpicture}
\begin{axis}[enlargelimits=false, axis on top, ticks=none, height=\jointPlotHeight,width=\jointPlotWidth]
\addplot graphics [xmin=-0.650000,xmax=1.500000,ymin=-0.200000,ymax=2.500000] {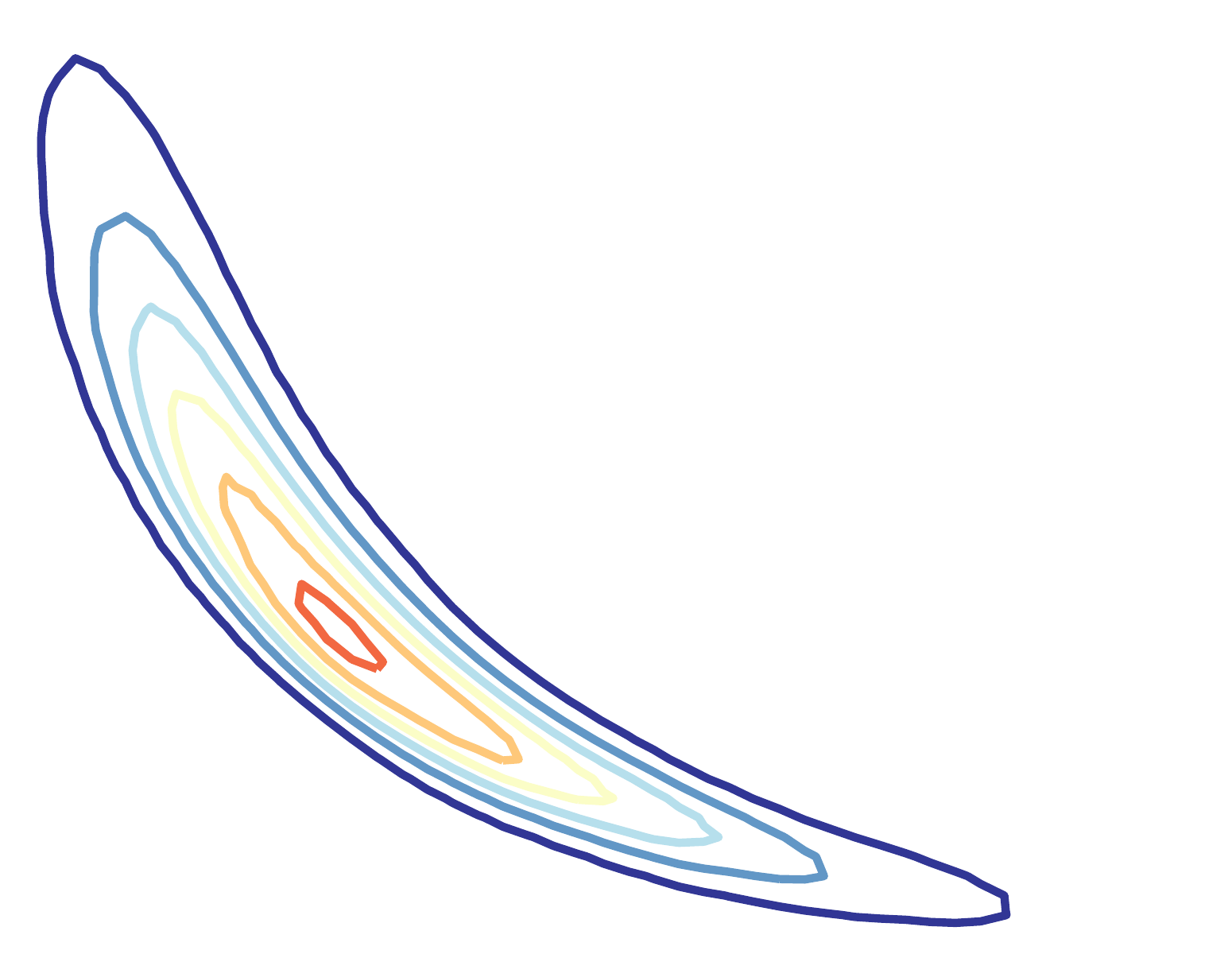};
\end{axis}
\end{tikzpicture}}\subfigure[degree $p=7$]{\begin{tikzpicture}
\begin{axis}[enlargelimits=false, axis on top, ticks=none, height=\jointPlotHeight,width=\jointPlotWidth]
\addplot graphics [xmin=-0.650000,xmax=1.500000,ymin=-0.200000,ymax=2.500000] {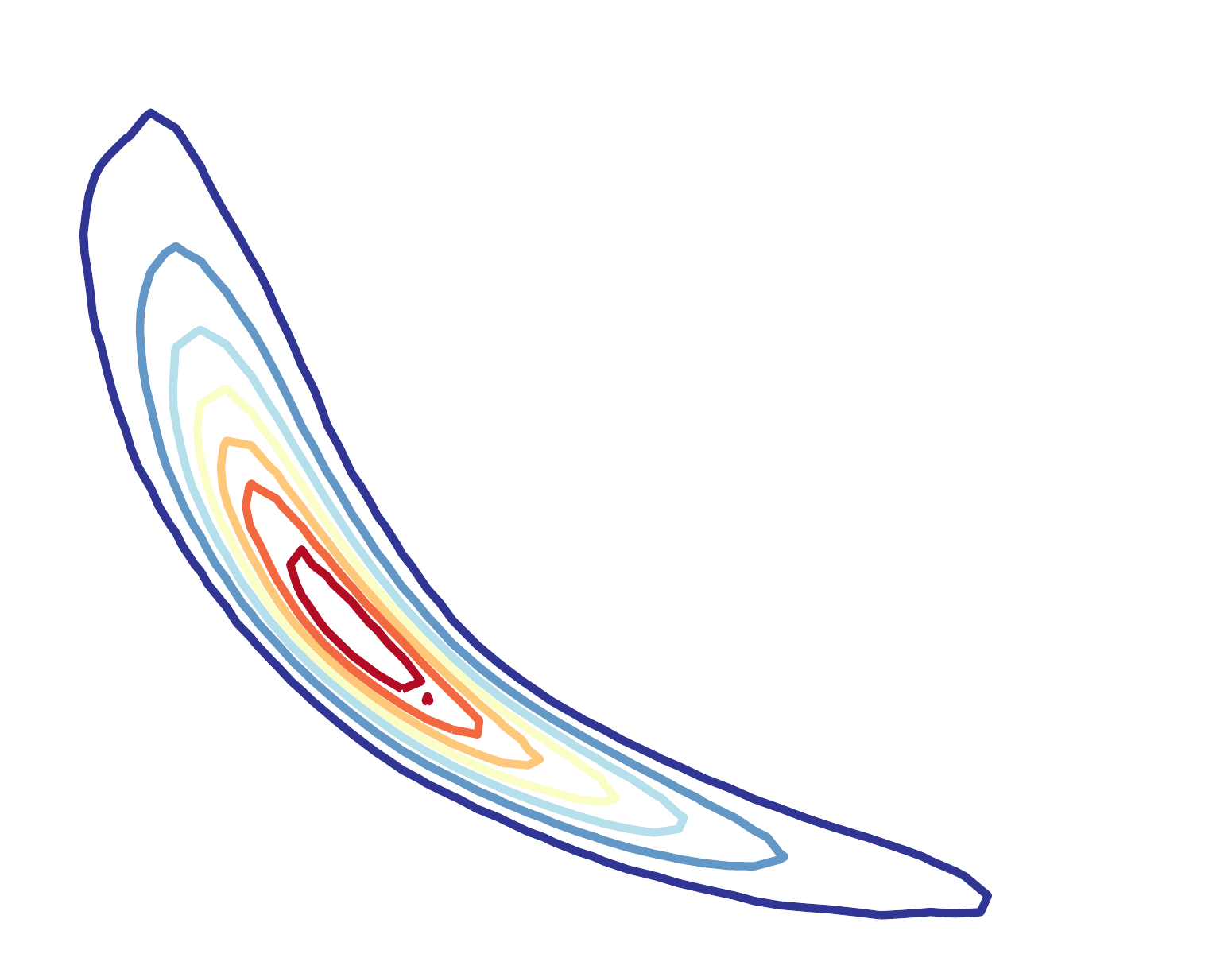};
\end{axis}
\end{tikzpicture}}

\caption[Approximate BOD posterior densities constructed with 50000 samples.]{
BOD problem of Section \ref{s:BODinverse} via inverse transport and conditioning.
Kernel density estimates of the conditional density $\ntd_{\paramt \vert \data=\datai}$, 
for $\datai=[0.18,   0.32,   0.42,   0.49,   0.54]$, 
using $50\, 000$ samples from either a reference adaptive MCMC sampler ({\it left})
or conditioned transport maps of varying total degree. Contour levels and color scales are constant for all figures.}
\label{fig:approxPost50000}
\end{figure}

\begin{table}[h]
\centering
\caption[Efficiency of offline inference maps on the BOD problem.]{
BOD problem of Section \ref{s:BODinverse} via inverse transport and conditioning. 
Efficiency of approximate Bayesian inference with inverse transport map from samples.
The  ``offline''  time is defined as the time it takes to 
compute an approximate inverse transport $\ifmap$ and a corresponding
approximate direct transport $\imap$ via regression (see Section \ref{sec:SampConst}).
The ``online'' time is the time required after observing 
$\{ \data = \datai \}$ to generate the equivalent of $30\, 000$ independent samples
from the conditional $\ntd_{\paramt \vert \data=\datai}$.
For MCMC, the ``online'' time is the average amount of time it takes to generate a chain with an effective sample size of $30\, 000$. 
}
\medskip
\footnotesize
\begin{tabular}{|c|c|r|r|r|}\hline
\multirow{2}{*}{map type} & \# training & \multicolumn{2}{c|}{Offline time (sec)} &  \multirow{2}{*}{Online time (sec)} \\
& samples & $\ifmap$ construction & $\imap$ regression &   \\\hline
\rowcolor{black!10} \multicolumn{2}{|c|}{MCMC ``truth''} &  \multicolumn{2}{c|}{NA} & 591.17 \\\hline
\multirow{2}{*}{$p=1$} & 5000 & 0.46 & 0.18 & 2.60\\
&  50000 &  4.55 & 1.65 & 2.32\\\hline
\rowcolor{blue!10} & 5000 & 4.13 & 1.36 & 3.54\\
\rowcolor{blue!10} \multirow{-2}{*}{$p=3$} &  50000 & 40.69 & 18.04 & 3.58\\\hline
\multirow{2}{*}{$p=5$} & 5000 & 22.82 & 8.40 & 5.80\\
&  50000 & 334.25 & 103.47 & 6.15\\\hline
\rowcolor{blue!10}& 5000 & 145.00 & 40.46 & 8.60\\
\rowcolor{blue!10}  \multirow{-2}{*}{$p=7$}  &  50000 & 1070.67 & 432.95 & 8.83\\\hline
\end{tabular}
\label{tab:off:time}
\end{table}

\subsection{Direct transport: map from densities}
\label{s:BODdirect}
We now focus on the computation of a direct transport as described in Section \ref{sec:DensConst}, using evaluations of an unnormalized target density.  Our goal here is to characterize a monotone increasing triangular transport map $T$ that pushes forward a standard normal reference density $\rd$ to the posterior distribution $\ntd_{\paramt \vert \data=\datai}$, i.e., the distribution of the BOD model parameters conditioned on a realization of the data $\datai$. We use the same realization of the data as in Section \ref{s:BODinverse}.

Figure \ref{fig:2dreftar} shows the reference and target densities. Notice that the target density exhibits a nonlinear dependence structure. This type of locally varying correlation can make sampling via standard MCMC methods (e.g., Metropolis-Hastings schemes with Gaussian proposals) quite challenging. In this example, the log-target density can be evaluated analytically up to a normalizing constant but direct sampling from the target distribution is impossible.  Thus it is an ideal setting for the computation of the direct transport as the minimizer of the optimization problem \eqref{OptimDirect}. We just need to specify how to approximate integration with respect to the reference measure in the objective of \eqref{OptimDirect}, and to choose a finite dimensional approximation space $\spaceMapT^h \subset \spaceMapT$ for the triangular map. Here we will approximate integration with respect to the reference measure using a tensor product of ten-point Gauss-Hermite quadrature rules. For this two-dimensional problem, this corresponds to $100$ integration nodes 
(i.e., $M=100$ in \eqref{OptimDirectApprox}). For $\spaceMapT^h$ we will employ a total-degree space of multivariate Hermite polynomials. In particular, we focus on third- and fifth-degree maps. The monotonicity constraint in \eqref{OptimDirect} is discretized pointwise at the integration nodes as explained in Section \ref{s:mapMonotone}. 

\begin{figure}[h]
\centering
\subfigure[Reference density]{\begin{tikzpicture}
\begin{axis}[enlargelimits=false, axis on top, ticks=none, height=\jointPlotHeight,width=\jointPlotWidth]
\addplot graphics [xmin=-2.200000,xmax=2.200000,ymin=-2.200000,ymax=2.200000] {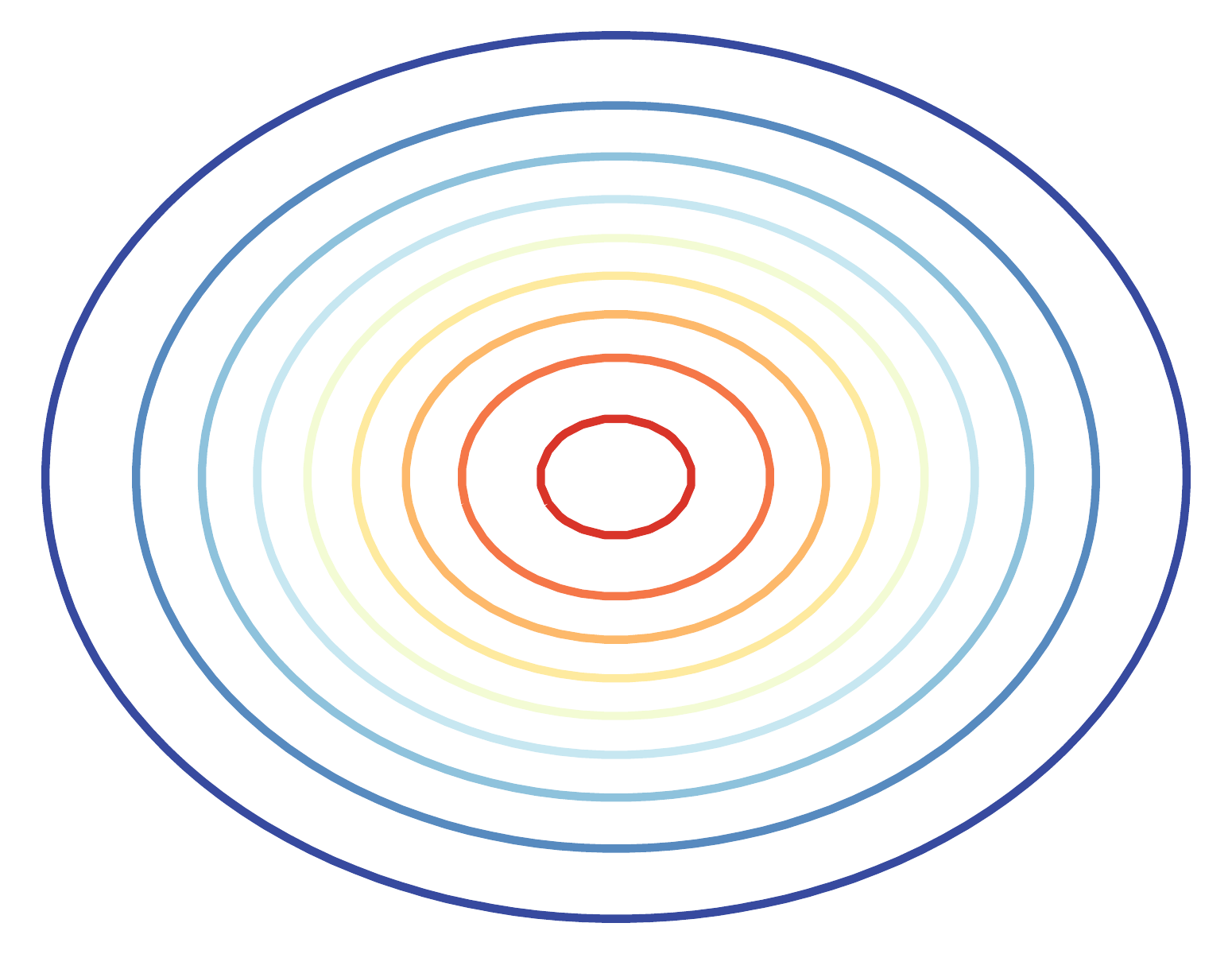};
\end{axis}
\end{tikzpicture}}
\subfigure[Target density]{\begin{tikzpicture}
\begin{axis}[enlargelimits=false, axis on top, ticks=none, height=\jointPlotHeight,width=\jointPlotWidth]
\addplot graphics [xmin=-0.650000,xmax=1.500000,ymin=-0.200000,ymax=2.500000] {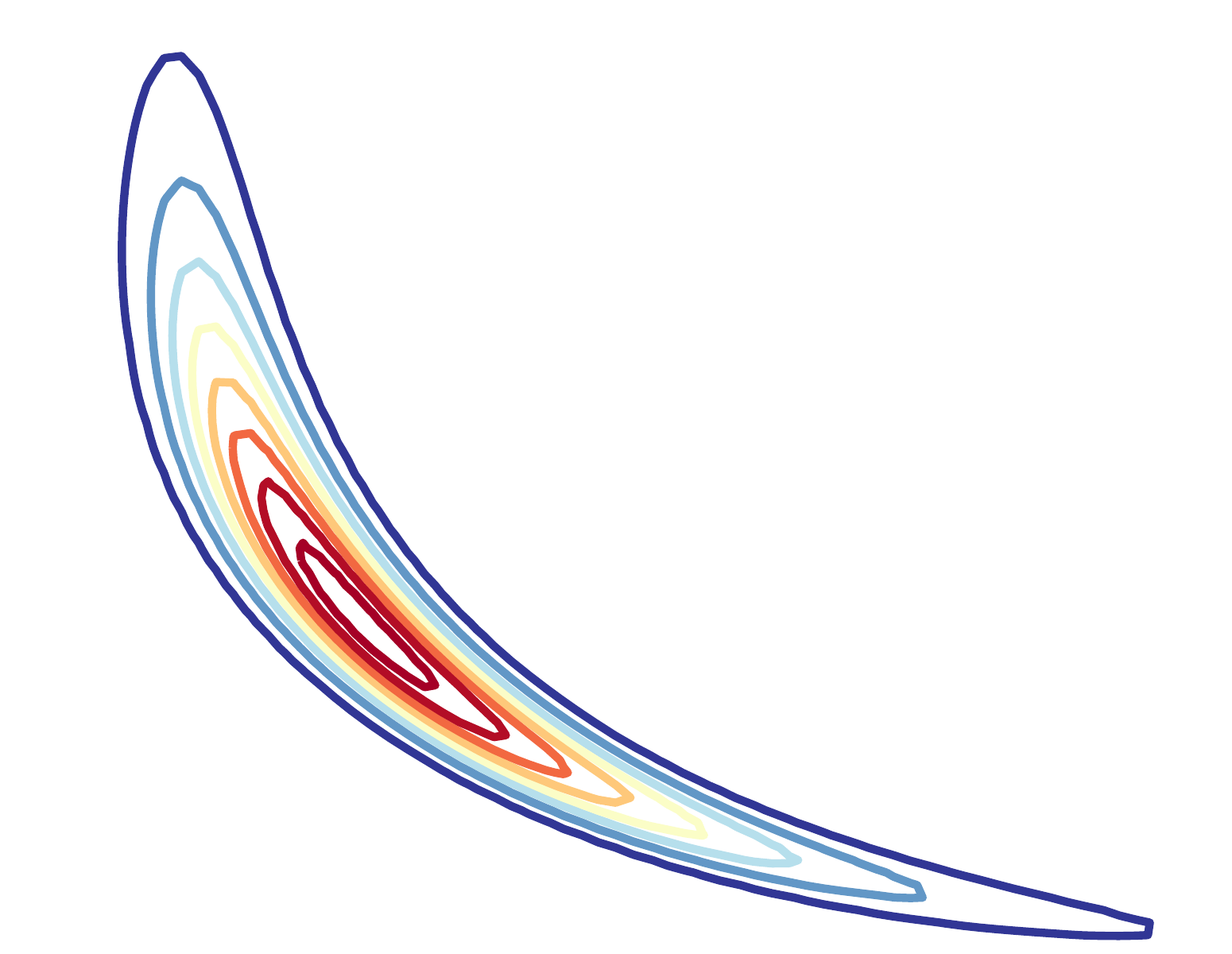};
\end{axis}
\end{tikzpicture}}
\caption[Reference and target densitied BOD2]{
BOD problem of Section \ref{s:BODdirect} via direct transport. 
Observations are taken at times $t\in \{1,2,3,4,5\}$. The observed data vector is given by $\datai=[0.18;   0.32;   0.42;   0.49;   0.54]$.}
\label{fig:2dreftar}
\end{figure}

Figure \ref{fig:BOD2pushforwards} shows the results of solving the discretized optimization problem \eqref{OptimDirectApprox} for the transport map. In particular, we show the pushforward of the reference density through the transport maps found by solving \eqref{OptimDirectApprox} for different truncations of $\spaceMapT^h$. As we can see from Figure \ref{fig:BOD2pushforwards}, an excellent approximation of the target density is already achieved with a degree-three map  (see Figure \ref{fig:BOD2pushforwards}(b)). This approximation improves and is almost indistinguishable from the true target density for a degree-five map (see Figure \ref{fig:BOD2pushforwards}(c)).
Thus if we were to compute posterior expectations using the approximate map as explained in Section \ref{s:bias}, we would expect virtually zero-variance estimators with extremely small bias. Moreover, we can estimate this bias using \eqref{biasBound}. 

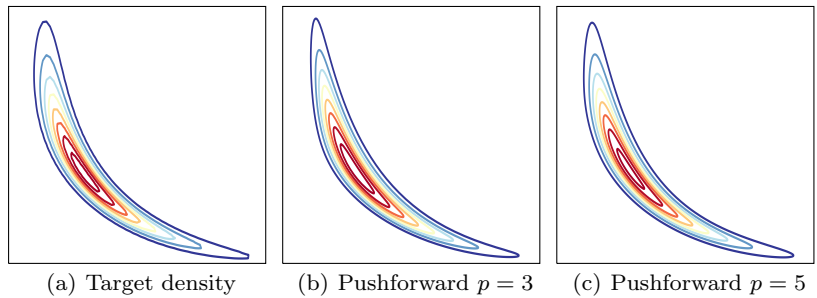
\begin{figure}[h] 
\centering
\subfigure[Target density]{\begin{tikzpicture}
\begin{axis}[enlargelimits=false, axis on top, ticks=none, height=\jointPlotHeight,width=\jointPlotWidth]
\addplot graphics [xmin=-0.650000,xmax=1.500000,ymin=-0.200000,ymax=2.500000] {figures/bodPlots/true_target_.pdf};
\end{axis}
\end{tikzpicture}}
\subfigure[Pushforward $p=3$]{\begin{tikzpicture}
\begin{axis}[enlargelimits=false, axis on top, ticks=none, height=\jointPlotHeight,width=\jointPlotWidth]
\addplot graphics [xmin=-0.650000,xmax=1.500000,ymin=-0.200000,ymax=2.500000] {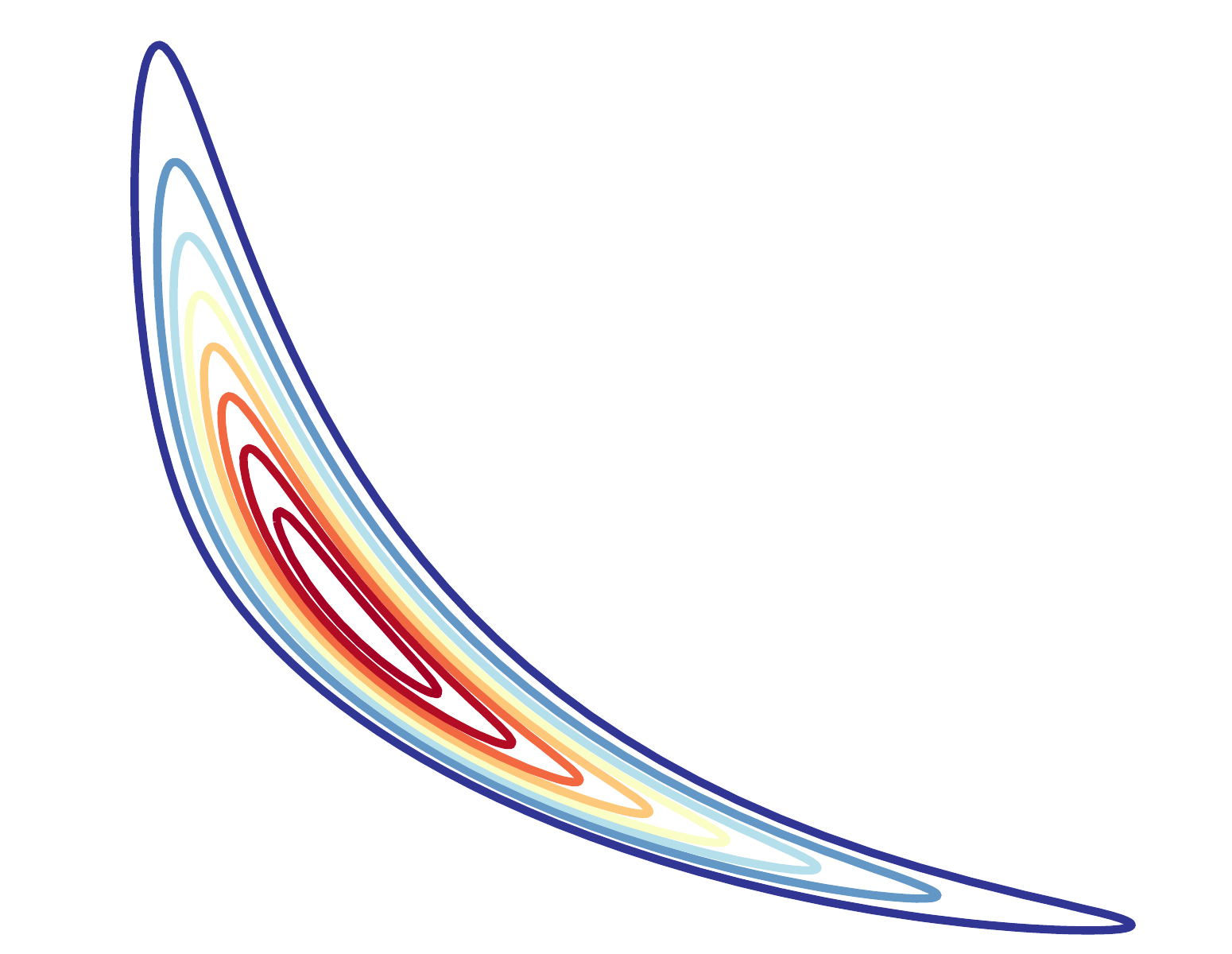};
\end{axis}
\end{tikzpicture}}
\subfigure[Pushforward $p=5$]{\begin{tikzpicture}
\begin{axis}[enlargelimits=false, axis on top, ticks=none, height=\jointPlotHeight,width=\jointPlotWidth]
\addplot graphics [xmin=-0.650000,xmax=1.500000,ymin=-0.200000,ymax=2.500000] {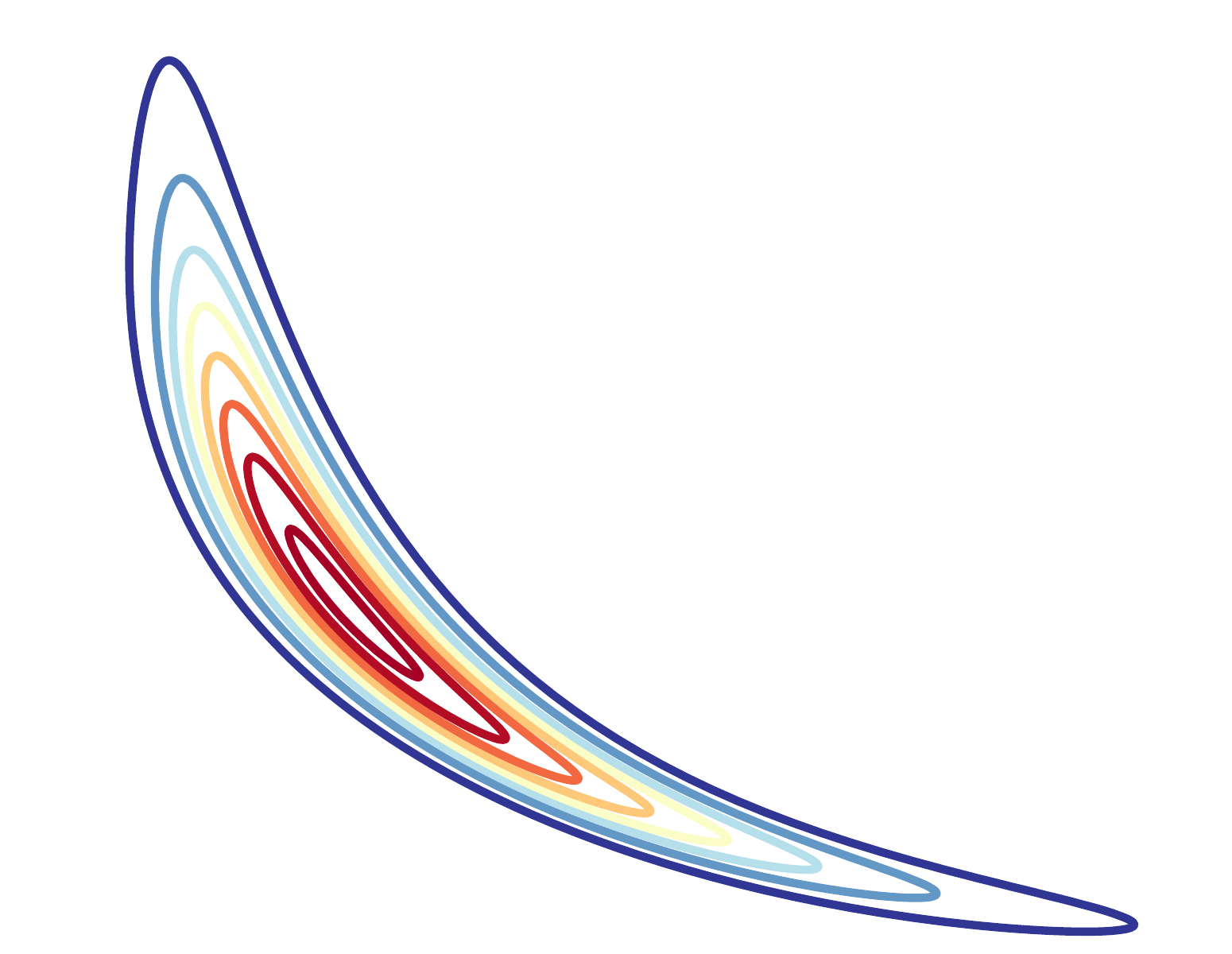};
\end{axis}
\end{tikzpicture}}
\caption[ push]{
BOD problem of Section \ref{s:BODdirect} via direct transport: pushforwards of the reference density under a given total--degree triangular map. The basis of the map consists of multivariate Hermite polynomials.
 The expectation with respect to the reference measure is approximated with a full tensor product Gauss-Hermite quadrature rule. The approximation is already excellent with a map of degree $p=3$.
}
\label{fig:BOD2pushforwards} 
\end{figure}

If one is not content with the bias of these estimators, then it is always possible to rely on asymptotically exact sampling of the pullback of the target distribution through the approximate map via, e.g., MCMC (see Section \ref{s:bias}). Figure \ref{fig:BOD2pullbacks} illustrates such pullback densities for different total-degree truncations of the polynomial space $\spaceMapT^h$. As we can see from Figure \ref{fig:BOD2pullbacks}(b,c), the pullback density is progressively ``Gaussianized'' as the degree of the transport  map increases. In particular, these pullbacks do not have the complex correlation structure of the original target density and are amenable to efficient sampling; for instance, even a Metropolis independence sampler \cite{RobertBook2004} could be very effective. Thus, approximate transport maps can effectively precondition and improve the efficiency of existing sampling techniques.

\begin{figure}[h]
\centering
\subfigure[Reference density]{\begin{tikzpicture}
\begin{axis}[enlargelimits=false, axis on top, ticks=none, height=\jointPlotHeight,width=\jointPlotWidth]
\addplot graphics [xmin=-2.200000,xmax=2.200000,ymin=-2.200000,ymax=2.200000] {figures/bodPlots/true_reference_.pdf};
\end{axis}
\end{tikzpicture}}
\subfigure[Pullback $p=3$]{\begin{tikzpicture}
\begin{axis}[enlargelimits=false, axis on top, ticks=none, height=\jointPlotHeight,width=\jointPlotWidth]
\addplot graphics [xmin=-2.200000,xmax=2.200000,ymin=-2.200000,ymax=2.200000] {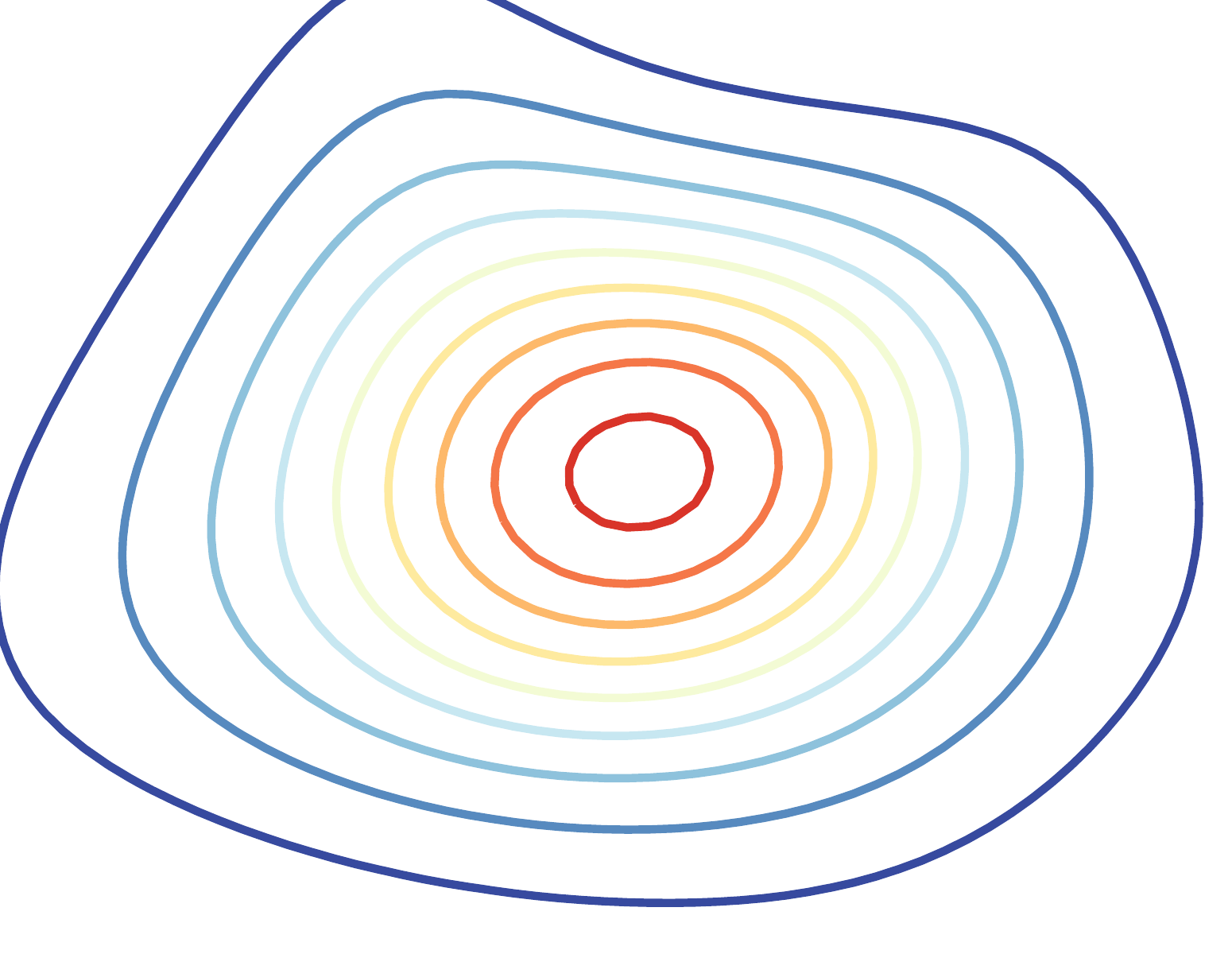};
\end{axis}
\end{tikzpicture}}
\subfigure[Pullback $p=5$]{\begin{tikzpicture}
\begin{axis}[enlargelimits=false, axis on top, ticks=none, height=\jointPlotHeight,width=\jointPlotWidth]
\addplot graphics [xmin=-2.200000,xmax=2.200000,ymin=-2.200000,ymax=2.200000] {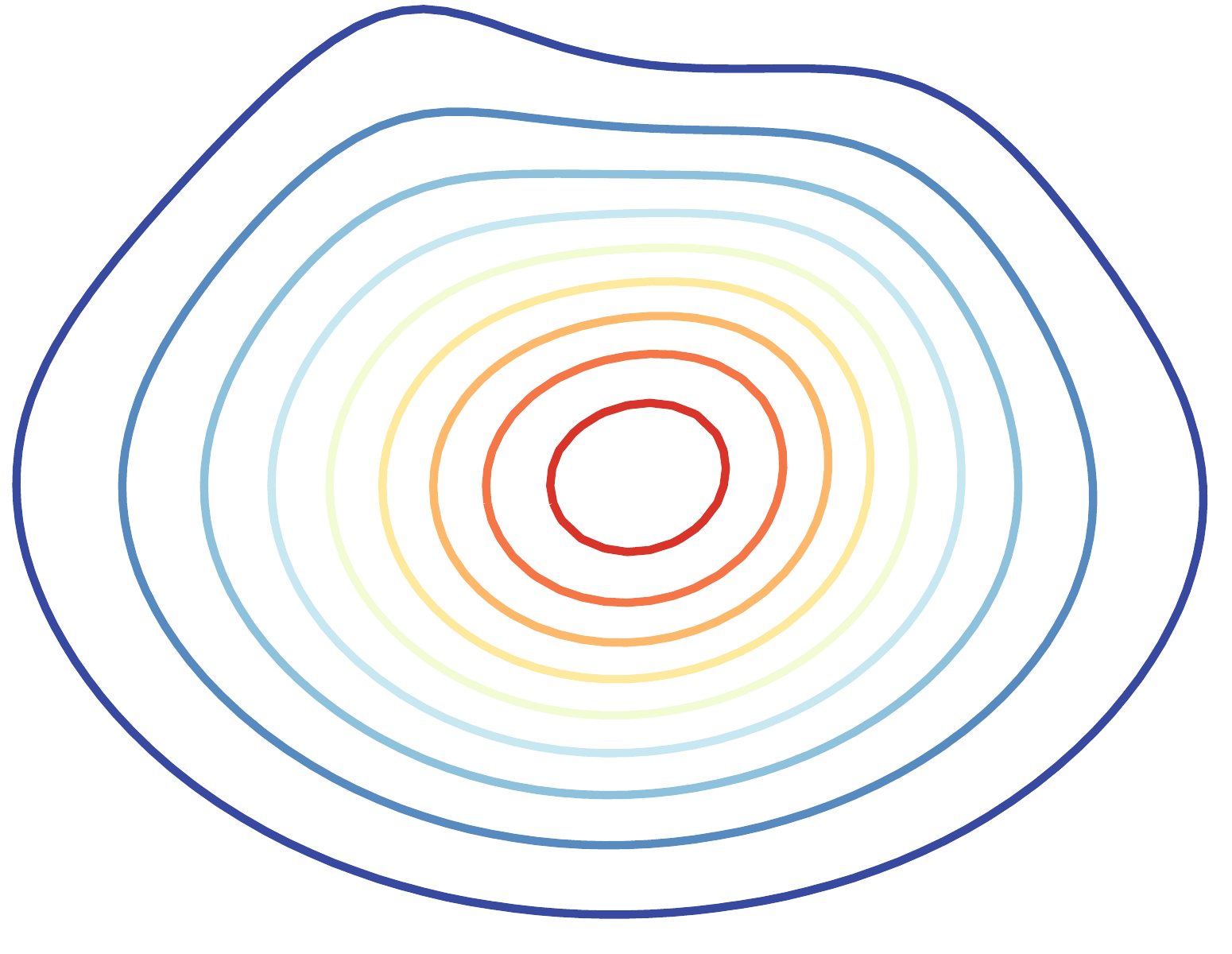};
\end{axis}
\end{tikzpicture}}
\caption[ ]{
BOD problem of Section \ref{s:BODdirect} via direct transport:
 pullbacks under a given total order triangular map of the target density. 
Same setup of the optimization problem as in Figure \ref{fig:BOD2pushforwards}.
The pullback density is progressively ``Gaussianized'' as the 
degree of the transport  map
increases. 
}
\label{fig:BOD2pullbacks} 
\end{figure}
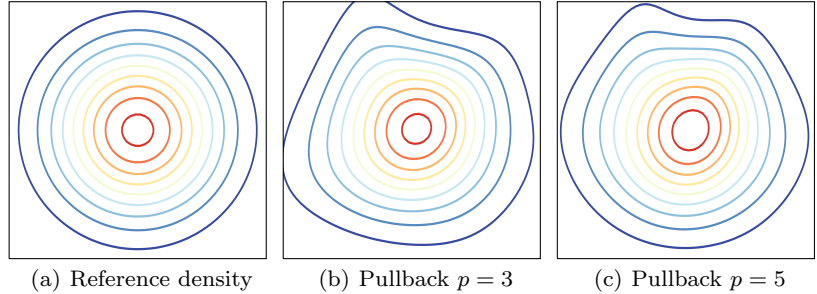

\section{Conclusions and outlook}
\label{sec:Conc}

In this chapter, we reviewed the fundamentals of the measure transport
approach to sampling.  The idea is simple but powerful. Assume that we
wish to sample a given, possibly non--Gaussian, target measure. We solve this problem by constructing a deterministic
transport map that pushes forward a reference measure to the target
measure. The reference can be any measure from which we can easily draw samples or 
construct quadratures (e.g., a standard Gaussian).  Under these
assumptions, pushing forward independent samples from the reference
through the transport map produces independent samples from the
target. This construction turns sampling into a trivial task: we only
need to evaluate a deterministic function. Of course, the challenge is
now to determine a suitable transport.  Though the existence of such
transports is guaranteed under weak conditions
\cite{Villani2009}, in this chapter we focused on target and reference
measures that are absolutely continuous with respect to the Lebesgue
measure, with smooth and positive densities.  These hypotheses make
the numerical computation of a continuous transport map particularly
attractive.  It turns out that a smooth triangular transport, the
Knothe--Rosenblatt rearrangement \cite{Rosenblatt1952,Carlier2010}, can
be computed via smooth and possibly unconstrained optimization.

To compute this transport, we considered two different scenarios.  In
Section \ref{sec:DensConst} we addressed the computation of a monotone
triangular \textit{direct transport}---a transport map that pushes
forward a reference measure to the target measure---given only the
ability to evaluate the unnormalized target density
\cite{Moselhy2011}.  This situation is very common in the context of
Bayesian inference.  The direct transport can be computed by solving a
smooth optimization problem using standard gradient-based techniques.  In Section
\ref{sec:SampConst}, on the other hand, we focused on a setting where
the target density is unavailable and we are instead given only
finitely many samples from the target distribution.  This scenario
arises, for instance, in density estimation \cite{Tabak2013}
or Bayesian inference with intractable likelihoods
\cite{wilkinson2011stochastic,marin2012approximate,Csillery2010}.  In
this setting, we showed that a monotone triangular \textit{inverse
transport}---a transport map that pushes forward the target measure to
the reference measure---can be computed efficiently via separable
convex optimization.  The direct transport can then be evaluated by
solving a nonlinear triangular system via a sequence of
one-dimensional root-findings (see Section \ref{sec:SampMap:inverse}).
Moreover, we showed that characterizing the target distribution as the
pushforward of a triangular transport enables efficient sampling from
particular conditionals (and of course any marginal) of the target
(see Section \ref{sec:CondSamp}). This feature can be extremely useful
in the context of online Bayesian inference, where one is concerned
with fast posterior computations for multiple realizations of the data
(see Section \ref{s:BODinverse}).

\medskip Ongoing efforts aim to expand the transport map framework by:
(1) understanding the fundamental \textit{structure} of transports,
and how this structure flows from certain properties of the target
measure; (2) developing rigorous and automated methods for the
adaptive refinement of maps; and (3) coupling these methods with more
effective parameterizations and computational approaches.

An important preliminary issue, which we discussed briefly in
Section~\ref{s:mapMonotone}, is how to enforce {\bf monotonicity} and
thus invertibility of the transport. In general, there is no easy way
to parameterize a monotone map. However, as shown in Section
\ref{s:mapMonotone} and detailed in \cite{bigoni2016monotone}, if we restrict our
attention to triangular transports---that is, if we consider the
computation of a Knothe--Rosenblatt rearrangement---then the
monotonicity constraint can be enforced {\it strictly} in the
parameterization of the map.  This result is inspired by monotone
regression techniques \cite{ramsay1998estimating} and is useful in the
transport map framework as it removes explicit monotonicity
constraints altogether, enabling the use of unconstrained
optimization techniques.

Another key challenge is the need to construct low-dimensional parameterizations of transport maps in high-dimensional settings. The critical observation in \cite{spantini16markov} is that Markov properties---i.e., the {conditional independence} structure---of the target distribution induce an intrinsic low dimensionality of the transport map in terms of \textbf{sparsity}  and \textbf{decomposability}. A sparse transport is a multivariate map where each component is only a function of few input variables, whereas a decomposable transport is a map that can be written as the {\it exact} composition of a finite number of simple functions. The analysis in \cite{spantini16markov} reveals that these sparsity and decomposability properties can be predicted {\it before} computing the actual transport simply by examining the Markov structure of the target distribution. These properties can then be explicitly enforced in the parameterization of candidate transport maps, leading to optimization problems of considerably reduced dimension. Note that there is a constant effort in applications to formulate probabilistic models of phenomena of interest using sparse Markov structures; one prominent example is multiscale modeling \cite{Parno2015}. 
A further source of low dimensionality in transports is \textbf{low-rank structure}, i.e., situations where a map departs from the identity only on a low-dimensional subspace of the input space \cite{spantini16markov}. This situation is fairly common in large-scale Bayesian inverse problems where the data are informative, relative to the prior, only about a handful of directions in the parameter space \cite{Cui2014lis,spantini2014optimal}.

Building on these varieties of low-dimensional structure, we still need to construct {\it explicit} representations of the transport. In this chapter, we have opted for a parametric paradigm, seeking the transport map within a finite-dimensional approximation class. Parameterizing high-dimensional functions is broadly challenging (and can rapidly become intractable), but exploiting the sparsity, decomposability, and low-rank structure of transports can dramatically reduce the burden associated with explicit representations. Within any structure of this kind, however, we would still like to introduce the fewest degrees of freedom possible: for instance, we may know that a component of the map should depend only on a small subset of the input variables, but what are the best basis functions to capture this dependence?  A possible approach is the \textbf{adaptive enrichment} of the approximation space of the map during the optimization routine. The main question is how to drive the enrichment. A standard approach is to compute the gradient of the objective of the optimization problem over a slightly richer approximation space and to detect the new degrees of freedom that should be incorporated in the parameterization of the transport. This is in the same spirit as adjoint--based techniques in adaptive finite element methods for differential equations \cite{bangerth2013adaptive}. In the context of transport maps, however, it turns out that one can  {\it exactly} evaluate the first variation of the objective over an infinite-dimensional function space containing the transport. A rigorous and systematic analysis of this first variation can guide targeted enrichment of the approximation space for the map \cite{bigoni2016monotone}. Alternatively, one could try to construct rather complex transports by \textbf{composing} simple maps and rotations of the space. This idea has proven successful in high-dimensional applications (see \cite{Parno2014thesis} for the details of an algorithm, and \cite{Tabak2013,Laparra2011} for related approaches). 

Even after finding efficient parameterizations that exploit available
low-dimensional structure, we must still search for the best transport
map within a finite-dimensional approximation space. As a result, our
transports will in general be only approximate. This fact should not
be surprising or alarming. It is the same issue that one faces, for
instance, when solving a differential equation using the finite
element method \cite{strang1973analysis}. The important feature of the
transport map framework, however, is that we can estimate the quality
of an approximate transport and decide whether to enrich the
approximation space to improve the accuracy of the map, or to accept
the bias resulting from use of an approximate map to sample the target
distribution. In Section \ref{s:bias} we reviewed many properties and
possible applications of approximate transports.  Perhaps the most
notable is the use of approximate maps to precondition existing
sampling techniques such as MCMC. In particular, we refer to
\cite{Parno2015mcmc} for a use of approximate transport maps in the
context of adaptive MCMC, where a low-order map is learned from MCMC
samples and used to construct efficient non-Gaussian proposals that
allow long-range global moves even for highly correlated targets.

So far, the transport map framework has been deployed successfully in
a number of challenging applications: high-dimensional non-Gaussian
Bayesian inference involving expensive forward models
\cite{Moselhy2011}, multiscale methods for Bayesian inverse problems
\cite{Parno2015}, non-Gaussian proposals for MCMC algorithms
\cite{Parno2015mcmc}, and Bayesian optimal experimental design
\cite{huan16maps}.
Ongoing and future applications of the framework include sequential
data assimilation (Bayesian filtering and smoothing), statistical
modeling via non-Gaussian Markov random fields, density estimation and
inference in likelihood-free settings (e.g., with radar and image
data), and rare event simulation.

\bibliographystyle{spmpsci}
\bibliography{MapChapter,ymmextrarefs}

\end{document}